\documentclass[11pt,a4paper]{article}

\usepackage{stile}
\usepackage{xcolor}
\usepackage{geometry}
\usepackage{amssymb}
\usepackage{bussproofs}

\usepackage{amsthm} %teoremi, definizioni, etc.
\usepackage{amsfonts}
\usepackage{verbatim}

\newtheorem{example}{Example}[section]

\usepackage{booktabs} %tabelle
\usepackage{caption}
\usepackage{graphicx}
\usepackage{multirow}
\usepackage{multicol}
\usepackage{yfonts}
\usepackage {tikz}
\usetikzlibrary{fit,shapes.geometric, arrows}
\usetikzlibrary{positioning}
\usetikzlibrary{calc,shapes.multipart,chains}
\definecolor {processblue}{cmyk}{0.96,0,0,0}

\usepackage{amsmath,trimclip}
\usepackage{subcaption}
\usepackage{mathrsfs}
\usepackage{hyperref}

\makeatletter
\DeclareRobustCommand{\Bot}{%
  \mathord{\vphantom{\bot}\mathpalette\mich@Bot\relax}%
}
\newcommand{\mich@Bot}[2]{%
  \ooalign{%
    $\m@th#1\bot$\cr
    \clipbox*{0pt 0pt {\width} {.5\height}}{\raisebox{.2\height}{$\m@th#1\bot$}}\cr
  }%
}
\makeatother

\def\ci{\perp\!\!\!\perp}
\title{Trustworthiness Preservation by Copies of Machine Learning Systems} %% Article title

\author{Leonardo Ceragioli and Giuseppe Primiero\footnote{LUCI Lab, Department of Philosophy, Università degli Studi di Milano}} %% Author name

\date{}

\begin{document}

\maketitle




%% Abstract
\begin{abstract}
A common practice of ML systems development concerns the training of the same model under different data sets, and the use of the same (training and test) sets for different learning models. The first case is a desirable practice for identifying high quality and unbiased training conditions. The latter case coincides with the search for optimal models under a common dataset for training. These differently obtained systems have been considered akin to copies. In the quest for responsible AI, a legitimate but hardly investigated question is how to verify  that trustworthiness is preserved by copies. In this paper we introduce a calculus to model and verify probabilistic complex queries over data and define four distinct notions: Justifiably, Equally, Weakly and Almost Trustworthy which can be checked analysing the (partial) behaviour of the copy with respect to its original. We provide a study of the relations between these notions of trustworthiness, and how they compose with each other and under logical operations. The aim is to offer a computational tool to check the trustworthiness of possibly complex systems copied from an original whose behavour is known.
\end{abstract}



\section{Introduction}

The problem of evaluating the trustworthiness of computational systems has received significant attention in the light of the new generation of Machine Learning algorithms \citep{10.1145/3491209,HENRIQUE2024100043,Benk-BENTYO-3}. For such systems, largely characterized by high complexity and opacity, the available approaches are \textit{ex-ante} and \textit{post-hoc} explanations \citep{RETZLAFF2024101243}. The latter, useful in particular in absence of a model to probe and investigate, consists of building transparent explaining agents able to emulate the behaviour of a given opaque system in order to evaluate the adherence of the latter to the former \citep{10.1007/978-3-030-88708-7_19}.

In this light, reasoning about the behaviour of such uncertain computational systems becomes crucial in order to 
verify their trustworthiness. This is the task of TPTND \citep{10.1093/logcom/exaf003}, a typed natural deduction calculus in which appropriate trust rules are defined according to different \textit{criteria}. 
A system can be evaluated by comparing its outputs with the theoretical \textit{expected} behavior. Or one can compare the output of the system with the \textit{ideal} measure for decision making, when available: \textit{e.g.} the objective probability that a candidate returns a loan, based on statistical information. The resulting distance obtains modulo a chosen threshold of admissibility. The calculus is equipped with a corresponding relational semantics \citep{KUBYSHKINA2024109212}, and it has been modeled to deal with different kinds of measures of entropy, in particular non-symmetric ones for the evaluation of divergent behaviours which may be characterised as bias \citep{DBLP:conf/aiia/PrimieroD22}. This approach has been implemented in terms of a verification platform for classifiers \cite{DBLP:conf/beware/CoragliaDGGPPQ23,coraglia2024evaluatingaifairnesscredit}. 
In the present work, our goal is to investigate in which cases trustworthiness, whatever the chosen criterion for its assessment, is preserved when the model is copied. While TPTND offers a formal tool for trust evaluation of computational processes in presence of a desirable or ideal behaviour, it remains limited in its expressivity for checking trustworthiness preservation under varying Training Sets for the same model, or for different models trained and tested on the same datasets.

\begin{example}\label{ex1}
    Consider an automatic credit scoring system used to evaluate the likelihood of defaulting on credit by assigning individual applicants a score. Falling below a certain threshold determines a negative output on the loan request. Individuals are characterized by properties such as gender, marital status, ethnicity, previous credit history and so on. The system is trained on historical data available to the financial institution containing the same properties as those of the intended Test Set. The result is a well-performing, robust model called $M1$, although it presents certain undersirable discrimination phenomena with respect to certain subgroups within ethnicity and maritual status. The financial institution decides to train on the same historical data a second model $M2$, characterised by differently tuned parameters and hyperparameters. Additionally, the original model is trained again on a different set of historical data provided by a local branch for deployment of the new model $M3$ in a different location. 
\end{example}

Example \ref{ex1}, albeit fictional, is close to a real case scenario. Credit scoring is actually performed by ML models trained with historical data and discrimination and bias are acknowledged issues \citep{DASTILE2020106263}. Moreover, the development of models like $M2$ to identify the better performing and most accurate one across a number of options trained on the same dataset is standard practice in several domains where AI is used. On the other hand, the development of models like $M3$, with retraining on more accurate datasets for local deployment is a desirable, although too often ignored, best practice. 

The practice of generating new models from the same training dataset, or to train the same model on different data can be considered akin to making copies.  An early analysis of this type of formal relations has been developed in view of technical artefacts and formal ontologies \citep{carrara2001identity,carrara2010copies}.In particular, when copies do not satisfy all properties of the original, one speaks of inexact or approximate copies of an original computational model \citep{DBLP:journals/logcom/AngiusP18}. Inexact and approximate copies preserve safety and liveness properties violations, while neither
of them is able to preserve their satisfaction. Viceversa, if safety and
liveness are satisfied by an inexact copy, those properties are satisfied by the copied system as well \citep{DBLP:journals/logcom/AngiusP23}. When these kinds of copy are considered in the context of machine learning systems and their inherently non-deterministic behaviour, the preservation of valid properties from the model to the copies requires appropriate weak morphisms, see \citep{manganiniprimieroforth}. These in turn necessarily translate to appropriate weak forms of trust, which are the focus of the present paper. Such analysis can also be useful in the development of digital twin technologies \citep{grieves2017digital}. They serve to test and understand how industrial infrastructure might behave during their lifecycle \citep{SEMERARO2021103469}. To this purpose, virtual environments and digital simulations are used, aided by data-driven algorithms. A major hurdle towards the efficient and safe deployment of processes realised by a digital twin in real production is represented by
the ability to formulate essential properties to be represented by the simulation and to verify that they are also correctly satisfied by
the real-world artefact. This is especially true when Machine Learning and Deep Learning  algorithms intervene in the
simulation process and huge amounts of data are used to draw correlations about potential behaviours of the system under analysis. Especially when such algorithms are adapted from ready-made models, the question whether their trustworthiness is preserved across new domains is crucial.

To allow for formal verification of these cases, a language is required that can express features and their values, equipped with frequencies. The evaluation of dependent probabilities is essential to express the value of a target feature for a given class and subclass of individuals. Equivalence relations can then be defined on formal systems of this form, modular with respect to either Training Set or model. We introduce the calculus \textit{Typed Natural Deduction for Probabilistic Queries} (TNDPQ for short) with the aim of reasoning about complex queries with probabilistic outputs and establishing  trustworthiness preservation under varying elements of the ML system, what we call its \textit{copies}. 
More generally, this can be seen as trustworthiness evaluation of an ML algorithm under variable input and model, and we shall consider 
different degrees in which trustworthiness is preserved across such copies.

TNDPQ --  and its ancestor TPTND -- belong to the family of probabilistic deductive systems, which includes many variants. \citep{DBLP:conf/lics/BacciFKMPS18, DBLP:conf/icfp/BorgstromLGS16}  offer stochastic untyped $\lambda$-calculi which focus on denotational or operational semantics with random variables, including a uniform formulation offered in \citep{DBLP:conf/lics/AmorimKMPR21}. In these systems a probabilistic program has the ability to model sampling from distributions, thereby rendering program evaluation a probabilistic process. \citep{DBLP:conf/icfp/BorgstromLGS16} introduces a measure space on terms and define step-indexed approximations, where a sampling-based semantics of a term is defined as a function from a trace of random samples to a value. In TNDPQ our terms are individual data points of a sample and the probabilistic approximation is associated with their output value. The inferential structure of TNDPQ is similar to calculi with types or natural deduction systems.
\citep{pierro2020atypetheory} deals with judgements that do not have a context and uses a subtyping relation for the term reduction. \citep{bor17} introduces ``probabilistic sequents'' of the form $\Gamma \vdash^n \Delta$, which are interpreted as stating that the probability of $\Gamma \vdash \Delta$ is greater than, or equal to, a function of the natural number $n$. In \citep{borivcic2019sequent}, and closer to the natural deduction in \citep{bor16},  probabilistic reasoning is formalised through sequents of the form $\Gamma \vdash_{a}^{b} \Delta$, which are interpreted as empirical statements of the form ``the probability of $\Gamma \vdash \Delta$ lies in the interval $[a,b]$''. TNDPQ deals only with sharp probabilities on the output types rather than on the derivability relation.
\citep{ghilezan2018probabilistic} introduces the logic P$\Lambda_{\rightarrow}$ where it is possible to mix standard Boolean and probabilistic formulas, the latter formed starting by a ``probabilistic operator'' $P_{\geq s} M:\sigma$ stating that the probability of $M:\sigma$ is equal to or greater than $s$. \citep{DBLP:conf/types/Adams015} introduces a quantitative logic with fuzzy predicates and conditioning of states to compute conditional probabilities. \citep{DBLP:journals/corr/Warrell16} offers a Probabilistic Dependent Type System (PDTS) via a functional language based on a subsystem of intuitionistic type theory including dependent sums and products, expanded to include stochastic functions. As a direct result of the Curry-Howard isomorphism, a probabilistic logic is derived from PDTS, and shown to provide a universal representation for finite discrete distributions. Unlike all these languages, the syntax of TNDPQ has been designed to express queries on a dataset returning probabilistic outputs and its inferential engine to account for closure under logical operators for such queries. The main advantage of the syntax is the modularity with respect to the elements denoting components of a ML systems, hence the ability to compare easily across modular changes (or copies).

The remaining of this paper is structured as follows.
In section~\ref{sec:ProofSystem} we start analyzing the components of an ML system and then we display a proof-system for calculus TNDPQ.
The proof-system is designed to decompose complex queries into logically simpler ones, and for simplicity it is presented (in subsections~\ref{subsec:Syntax} and \ref{subsec:ProofSystem}) in a language that keeps implicit the components of the ML system.
Then, the language of the calculus is expanded by making those components explicit in sub-subsection~\ref{subsub:ExtLang}.

In section~\ref{sec:TrustCopies}, we start to investigate the notion of trustworthy copy of an ML system.
In this and in the subsequent section, we restrict our focus only to ML systems with atomic target variables that receive only atomic values.
Moreover, trust is investigated for \textit{local} applications of ML systems to specific lists of values assignments $\sigma$ and for specific atomic target variables $a$.
In subsection~\ref{subsec:DifferentCopies}, we discriminate copies on the basis of the elements of the original system that are changed in the copy, while in subsection~\ref{subsec:DiffTrust} we discriminate them on the basis of the strength of the Trust relation that holds between original and copy.
As for the elements of the original that are changed in the copy, both Training Set and Learning Algorithm are considered.
As for the strength of the Trust relation, we distinguish between Justifiably, Weakly, Equally and Almost Trustworthy copies of an original.
The distinction is based on how many outputs are preserved and how similar they are between original and copy.

In section~\ref{sec:RelationsTrust}, the relations between the different notions of trust are investigated.
Since relations between notions of trustworthiness do not depend on which element of the original system is preserved and which one is changed in the copy, but only on the strength of such notions, in this section we use variables for trained systems, without discerning their components.
In subsection~\ref{subsec:entailment}, the relations of entailment between notions of trustworthiness are investigated, together with the complementary assumptions needed to increase the degree of trustworthiness of a copy.
In subsection~\ref{subsec:composition}, compositions of trustworthiness relations are investigated, researching whether and to what extent compositions of trustworthy relations are themselves trustworthy, and so how continual copying impacts trustworthiness.

In section~\ref{sec:TrustLogic}, we extend the study of trustworthiness to logical queries.
First, in subsection~\ref{subsec:TrustMultAppl}, trust is defined for general applications of ML systems to sets of lists of value attributions $\Sigma$ and to sets of atomic target variables $ \mathscr{A} $.
Then, in subsection~\ref{subsec:LogicTrust}, the definitions of Justifiably, Equally, Weakly, and Almost Trustworthy copy are extended to cover non-atomic target variables and values.
In particular, with these definitions in place, in sub-subsection~\ref{subsubsec:PreservationResults}, some results about the preservation of trustworthiness under logical construction and deconstruction of the queries are proved.

Section~\ref{sec:Conclusion} concludes and proposes some further steps of this investigation, while in~\ref{appendix} we present formal details of mutual exclusivity for values of variables, connected to the treatment of disjunction in section~\ref{sec:ProofSystem}.

\section{ML systems, proof-theoretically}
\label{sec:ProofSystem}

We start considering an ML system as composed of the following elements:

\begin{itemize}
\item a Learning Algorithm 
\item a Training Data Set, used to obtain the Learned Model
\item a Test Data Set, on which the Learned Model is applied to make predictions or classifications,
\item Outputs.
\end{itemize}

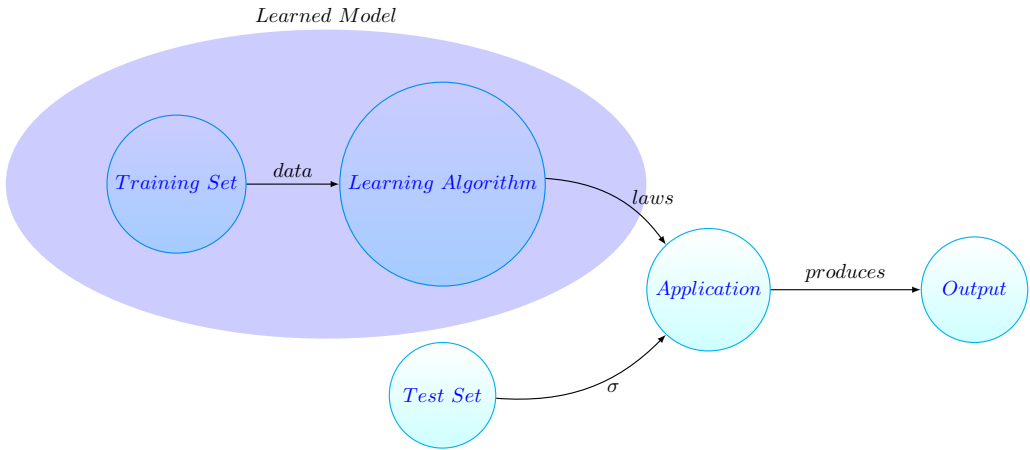
\begin{figure}
\scalebox{.70}{
\begin {tikzpicture}[-latex ,auto ,node distance =2 cm and 5cm ,on grid ,
semithick ,
state/.style ={ circle ,top color =white , bottom color = processblue!20 ,
draw,processblue , text=blue , minimum width =2 cm}]
\node[state] (A) {$Application$};
\node[state] (B) [above left=of A] {$Learning$ $Algorithm$};
\node[state] (C) [left =of B] {$Training$ $Set$};
\node[state] (D) [below left =of A] {$Test$ $Set$};
\node[state] (E) [right =of A] {$Output$};

\node[fit=(B)(C), ellipse, fill=blue, opacity=0.2, label=above: $Learned$ $Model$] {};

\path (C) edge node[above =0 cm] {$data$} (B);
\path (B) edge [bend left =25] node[above =0 cm, right =0.2 cm] {$laws$} (A);
\path (D) edge [bend right =25] node[above =0 cm, right =0.2 cm] {$\sigma$} (A);
\path (A) edge node[above =0 cm] {$produces$} (E);

\end{tikzpicture}
}
\caption{Components of an ML system and its application}
\label{fig:ontology}
\end{figure}

While the Test Set is usually not accounted as part of the system itself, it is nonetheless crucial to our purposes: our notion of copy is going to be the result of replacements on an original ML system of either the Training Set, or the Learning Algorithm, and the analysis is made based on their deployment on a Test Set. In the first case, a copy is a Learning Algorithm (function) differently trained, and possibly executed on the same Test Set. In the second case, the copy is a different algorithm trained on the same Test Set and possibly executed on the same Test Set. For each of these cases, it is plausible to ask whether the system's behavior is going to be the same, or at least sufficiently similar. And when the original system is evaluated as trustworthy, it is appropriate to ask whether the copy is going to preserve, and how much, such property.

In the following of this section, we introduce a language and proof system for these components.

\subsection{Syntax}
\label{subsec:Syntax}

The syntax of our calculus is generated by the following grammar in Backus-Naur form:

\begin{equation*}
\begin{split}
AtT & := a   \mid a_{n}  \\
T & := AtT \mid  
%t\mid u \mid v \mid z \mid 
\langle T,T \rangle \mid fst(T) \mid snd(T) \mid [T] T \\
AtO & := \alpha \mid \alpha^{n} \\
O & :=  AtO \mid O^{\bot} \mid O + O \\
pO & := O_{r} \mid (O \times O)_{r} \mid (O \rightarrow O)_{r} \\
VA & := AtT:O \\
pVA & := T:pO \\
lVA & := VA \mid lVA, VA \\
J & := lVA \rhd pVA \\
AtJ & := lVA \rhd AtT:AtO_{r} \\
\end{split}
\end{equation*}

The first lowercase letter of the Latin alphabet $a$ in the set $AtT$ (possibly indexed with natural numbers, $a_{i}$) is used for atomic variables for attributes.
$\mathscr{A}$ is used to denote a set of atomic variables.
Lowercase Latin letters from the \textit{end} of the alphabet $ t,u,v,z $ are used as variables to denote elements in $T$, comprising atomic variables and closed under logical construction.
$\mathscr{T}$ is used to denote a set of variables.
Given two variables $t$ and $u$, they can be composed in pairs $\langle t,u \rangle$ or conditionally $ [t] u $.
$fst(t)$ and $snd(t)$ are the first and second projections, used to deconstruct conjunctions, and are characterized by the usual evaluation rules:

\begin{equation*}
\begin{split}
fst(\langle t,u \rangle) & = t \\
snd(\langle t,u \rangle) & = u
\end{split}
\end{equation*}

\noindent
Note that we do not have a similar operation for the conditional, since any judgment containing a conditional is just a notational variation of a judgment containing no conditionals, as will be clear from its rules.

The first letter of the Greek alphabet $ \alpha $ (possibly with indices, $\alpha^{i}$; when necessary we will refer to the set of all indices of a given variable $a$ as $I$) from the set $AtO$ is used for possible atomic values of atomic variables.
The other lowercase Greek letters  $\beta , \delta , \gamma $ are used for possible deterministic values of atomic variables from the set $O$, which comprise atomic values and are closed under negation $\bot$ and disjunction $+$.
Lowercase Greek letters with real numbers as subscripts $\alpha _{r} , \beta _{r}, \delta _{r}, \gamma _{r}$ are used for probabilistic values for variables, which comprise atomic values and are closed under negation $\bot$, disjunction $+$, conjunction $ \times $ and implication $\rightarrow$.

The expression $a:\alpha$ stands for the atomic variable $a$ receiving the atomic value $\alpha$, and $a:\alpha _{r}$ stands for atomic variable $a$ receiving value $\alpha$ with probability $r \in \mathbb{R}$.
Moreover, since every atomic variable has a finite number of possible atomic values, we will write $a:\alpha ^{1}_{r_{1}}, \ldots , \alpha ^{n}_{r_{n}}$ to indicate that $\alpha ^{1}, \ldots , \alpha ^{n}$ are all the possible atomic values of atomic variable $ a $, and $a$ receives value $\alpha^{i}$ with probability $ r_{i} $.
The expression $a:\beta $ stands for the atomic variable $a$ receiving the value $\beta$, while the expression $t:\beta _{r}$ stands for the variable $t$ receiving the value $\beta$ with probability $ r $.
We will call ``values attributions'' ($VA$) events such as $t:\beta $ (variable $t$ receives value $\beta$), and ``probabilistic values attributions'' ($pVA$) events such as $t:\beta _{r}$ (variable $t$ receives value $\beta$ with probability $r$).

Note that, while values attributions require an atomic variable, probabilistic values attributions are more `liberal' and can be constructed using any kind of variable.
Accordingly, values ($O$) -- used in $VA$ -- comprise only atomic values, disjunctions and negation (corresponding to atomic variables), while probabilistic values ($pO$) -- used in $pVA$ -- are their proper extension, comprising also products and conditionals (corresponding to possibly non-atomic variables). Value attributions are collected in lists ($lVA$).
For compactness, we will use $\sigma _{1}, \ldots, \tau _{1}, \ldots$ to indicate such lists of value attributions and $\Sigma$ to indicate a set of lists of value attributions. A list of value attributions represents a single unit in the population of the Test Set. If more variables receive values, then the description of such unity of population is more specific.

\begin{example}[List of value attributions]
As an example,
\medskip
\[
Age: 27 ,\: Gen.: f,\: M.S.: married + divorced,\: Etn.: white^{\bot} 
\]

\medskip
\noindent
is used to refer to a single unit of population that satisfies the following value attributions: a married or divorced woman who is 27 years old and is not white.
\end{example}

In our analysis, the system receives in input a list of value attributions and gives a probabilistic value attribution for the target variable as its output. Intuitively, this corresponds to asking the system the probability that a subject described by the list of value attributions receives a given value for the target attribute.
We use judgments ($J$) to express such answers to the queries.
In particular, we write $\sigma \rhd u:\delta _{r}$ to mean that the system gives probability $r$ to the assignment of value $\delta$ to variable $u$ for subjects described by the list of value attributions $\sigma$.
Lastly, atomic judgments ($AtJ$) are a special kind of judgments of the form $\sigma \rhd a:\alpha _{r}$, where the probabilistic values attribution gives an atomic value for an atomic target variable.
Moreover, judgments of the form $\sigma \rhd a:\alpha^{1} _{r_{1}}, \ldots , \alpha ^{n}_{r_{n}}$ indicate that $\alpha ^{1}, \ldots , \alpha ^{n}$ are all the possible atomic values of atomic variable $ a $, and $a$ receives value $\alpha^{i}$ with probability $ r_{i} $ for subjects described by the list of value attributions $\sigma$.

\begin{example}[Judgment]
\label{exa:Jud}
As an example, the judgment
\medskip
\[
Age: 27 ,\: Gen.: f,\: M.S.: married + divorced,\: Etn.: white^{\bot} \rhd Loan : yes _{0.60}
\]

\medskip
\noindent
says that a married or divorced woman who is 27 years old and is not white has a probability of $0.65$ of receiving a loan.
\end{example}

\subsection{Proof System: TNDPQ}
\label{subsec:ProofSystem}

Our proof system TNDPQ contains one rule for atomic judgments and other rules for dealing with logically complex judgments.
The rule $AtQuery$ in table~\ref{tab:AtQuery} introduces atomic judgments, and serves as an axiom to import into the calculus the probabilistic laws that the ML system has learned from the Training Set, in the form of conditional probabilities.
Note that, according to our grammar, $\sigma$ contains only value attributions and not probabilistic value attributions, hence our dependence is from deterministic values only.
Hence, the conclusion of $AtQuery$ corresponds to the output of a query asking the probability that an individual described by the list of assignments $\sigma$ receives atomic value $\alpha$ for the target attribute corresponding to atomic variable $a$.

The rules in tables~\ref{tab:triangle} and \ref{tab:LRtriangle} are designed to investigate how logically complex questions are addressed by an ML system.
Note that, as stressed in the previous section, only atomic variables can occur on the left of $ \rhd $, while both atomic and complex variables can occur on its right.
Moreover, $\beta \times \delta$ and $\beta \rightarrow \delta$ require an introduction that builds complex variables from those corresponding to their constituents $ \beta $ and $ \delta $, as opposed to $\beta + \delta$ and $\beta ^{\bot} $, which applies to variables as simple as those corresponding to their constituents $ \beta $ and $ \delta $.
For this reason, table~\ref{tab:triangle}, which deals with the connectives that can occur only on the right of $ \rhd $, contains the rules for $\times$ and $ \rightarrow $.
On the other hand, table~\ref{tab:LRtriangle} deals with the connectives that can occur both on the left and on the right of $ \rhd $, and so contains the rules for $ + $ and $ \bot $.

Before analyzing the logical rules one by one, let us deal with their general structure.
As it is usual with natural deduction, logical rules are divided into introduction (I) and elimination (E) rules.
Moreover, there is a general procedure for defining the E-rules for a connective as a function of the corresponding I-rules:

\begin{definition}[Inversion Principle]
\label{def:InvPrinc}
Given a connective $\dagger$, for every I-rule

\begin{center}
\AxiomC{$\sigma _{1} \rhd t_{1} : \beta ^{1} _{g_{1}} $}
\AxiomC{$ \ldots $}
\AxiomC{$\sigma _{n} \rhd t_{n} : \beta ^{n} _{g_{n}} $}
\RightLabel{I$ \dagger $i}
\TrinaryInfC{$\tau \rhd u : \delta _{f(g_{1}, \ldots , g_{n})} $}
\DisplayProof
\end{center}

\noindent
there are $ n $ E-rules, each deriving a premise of the corresponding I-rule I$ \dagger $i from its conclusion and its other premises:

\begin{center}
\scalebox{.85}{
{\AxiomC{$\tau \rhd u : \delta _{h} $}
\AxiomC{$ \ldots $}
\AxiomC{$\sigma _{n} \rhd t_{n} : \beta ^{n} _{g_{n}} $}
\RightLabel{E$ \dagger $i$ _{1} $}
\TrinaryInfC{$\sigma _{1} \rhd u : \delta _{f^{1}(h, \ldots , g_{n})} $}
\DisplayProof}~{...}~{
\AxiomC{$\sigma _{1} \rhd t_{1} : \beta ^{1} _{g_{1}}$}
\AxiomC{$ \ldots $}
\AxiomC{$\tau \rhd u : \delta _{h}$}
\RightLabel{E$ \dagger $i$ _{n} $}
\TrinaryInfC{$\sigma _{n} \rhd u : \delta _{f^{n}(g_{1}, \ldots , h)} $}
\DisplayProof}}
\end{center}

\noindent
Where, the functions $ f $ and $ f^{i} $ computing the probability of the conclusion of respectively the I-rule and the i-th E-rule from their premises are such that $ f^{i}(g_{1}, \ldots , f(g_{1}, \ldots , g_{n}), \ldots , g_{n}) = g_{i} $.
\end{definition}

\noindent
Note that, as a special case of this definition, when an I-rule has only one premise, the corresponding E-rule is obtained just by inverting the premise and the conclusion of the I-rule.

\begin{table} %è una tabella floating, ma va comunuqe scritta approssimativamente dove si vuole che compaia.
\caption{Rules for Atomic Query} %didascalia
\label{tab:AtQuery} 
\centering %meglio di center, per la spaziatura

\AxiomC{}
	\RightLabel{{\tiny AtQuery}}
	\UnaryInfC{$ \sigma \rhd a: \alpha _{g} $}
\DisplayProof

\end{table}

\begin{table} %è una tabella floating, ma va comunuqe scritta approssimativamente dove si vuole che compaia.
\caption{Rules for $\rightarrow$ and $\times$ on the right} %didascalia
\label{tab:triangle} 
\centering %meglio di center, per la spaziatura

\begin{tabular}{c c}

\multicolumn{2} {c} {
\AxiomC{$ \sigma , t: \beta \rhd u: \delta _{g} $}
	\RightLabel{{\tiny I/E$\rightarrow$}}
	\doubleLine
	\UnaryInfC{$ \sigma  \rhd [t] u : \beta \rightarrow \delta _{g} $}
\DisplayProof
}

\\[0.8cm]

%\multicolumn{2} {c} {
\AxiomC{$ \sigma ,t: \beta  \rhd u : \delta _{g} $}
	\AxiomC{$ \sigma \rhd t: \beta _{f} $}
	\RightLabel{{\tiny I$\times$1}}
	\BinaryInfC{$ \sigma  \rhd \langle t, u \rangle: (\beta\times  \delta) _{f\cdot g} $}
\DisplayProof		
%}

&

\AxiomC{$ \sigma ,u: \delta \rhd t : \beta _{g} $}
	\AxiomC{$ \sigma \rhd u: \delta _{f} $}
	\RightLabel{{\tiny I$\times$2}}
	\BinaryInfC{$ \sigma  \rhd \langle t, u \rangle: (\beta\times  \delta) _{f\cdot g} $}
\DisplayProof

\\[0.8cm] 
\multicolumn{2} {c} {
\AxiomC{$ \sigma  \rhd t: (\beta\times  \delta) _{f} $}
	\AxiomC{$ \sigma , fst(t) : \beta \rhd  snd(t): \delta _{g\neq 0} $}
	\RightLabel{{\tiny E$\times$1a}}
	\BinaryInfC{$\sigma \rhd fst(t): \beta _{f/g}  $}
\DisplayProof
}

\\[0.8cm] 

\multicolumn{2} {c} {
\AxiomC{$ \sigma  \rhd t: (\beta\times  \delta) _{f} $}
	\AxiomC{$ \sigma \rhd fst(t): \beta _{g\neq 0} $}
	\RightLabel{{\tiny E$\times$1b}}
	\BinaryInfC{$ \sigma , fst(t) : \beta \rhd  snd(t): \delta _{f/g} $}
\DisplayProof
}

\\[0.8cm]

\multicolumn{2} {c} {
\AxiomC{$ \sigma  \rhd t: (\beta\times  \delta) _{f} $}
	\AxiomC{$ \sigma , snd(t) : \delta \rhd fst(t): \beta _{g\neq 0} $}
	\RightLabel{{\tiny E$\times$2a}}
	\BinaryInfC{$\sigma \rhd snd(t): \delta _{f/g}  $}
\DisplayProof
}

\\[0.8cm] 

\multicolumn{2} {c} {
\AxiomC{$ \sigma  \rhd t: (\beta\times  \delta) _{f} $}
	\AxiomC{$ \sigma \rhd snd(t): \beta _{g\neq 0} $}
	\RightLabel{{\tiny E$\times$2b}}
	\BinaryInfC{$ \sigma , snd(t): \delta \rhd fst(t): \beta _{f/g} $}
\DisplayProof
}
\\
\end{tabular} 
\end{table}

Let us now deal with the specific rules, starting with those in table~\ref{tab:triangle}.
Rules I$\rightarrow$ and E$\rightarrow$ deal with implication, treated as conditional probability.
Since $AtQuery$ introduces judgments that are essentially conditional probabilities, the introduction of implication corresponds to residuation.
In other words, in order to find the probability of 

\begin{center}
``$u$ receives value $\delta$, if $t$ receives value $\beta$'' 
\end{center}

\noindent
for an individual described by the list $\sigma$, we just use $t:\beta$ in the query alongside  $\sigma$ and then move it on the right. Note that, since I$\rightarrow$ has only one premise, by definition~\ref{def:InvPrinc}, I$\rightarrow$ and E$\rightarrow$ characterize $\rightarrow$-judgments as a purely notational variance of atomic judgments.
This is possible because, as already stressed, atomic judgments themselves correspond to conditional probabilities.
Note however that the judgments 
\medskip
\[
\sigma , t: \beta \rhd u: \delta _{g} \qquad and \qquad  \sigma  \rhd [t] u : \beta \rightarrow \delta _{g}
\]

\medskip
\noindent
correspond to conceptually different queries.

\begin{example}[Atomic and conditional judgments]
\label{exa:Cond}
The following judgement
\medskip
\[
Age: 27 ,\: M.S.: married + divorced,\: Etn.: white^{\bot} \rhd  [ Gen. ] Loan : f \rightarrow yes _{0.60}
\]

\medskip
\noindent
queries the probability that an individual satisfying the premises (a 27 years old person who is married or divorced and not white) receives the loan, if she is a woman.
Although they have the same probability, the query is different from the one considered in example~\ref{exa:Jud}: in this case, we are asking a conditional probability, while in the previous example we were asking the probability of a value attribution to an atomic variable.
\end{example}

The rules I$\times$1 and I$\times$2 use conditional probability to introduce conjunction.
They rely on the property of probabilities:
\medskip
\[
Pr(B | A)Pr(A) = Pr(A\times B) = Pr(A | B)Pr(B)
\]

\medskip
\noindent
Where $Pr(B | A)$ stands for the probability of $B$ conditioned on the assumption that $ A $ holds, and $Pr(A\times B)$ stands for the probability that both $ A $ and $ B $ hold.
There are two such rules, because since conjunction is symmetric, $\langle t, u \rangle: (\beta\times  \delta)$ can be introduced using either $t : \beta \rhd u : \delta $ together with $t : \beta$, or using $ u : \delta \rhd t : \beta $ together with $u :\delta$.\footnote{
We assume that $t \neq u$.
Note that no interesting case of conjunction is dropped by this restriction, while the system behaves much better.
As an example, see~\ref{appendix}.
}

Following definition~\ref{def:InvPrinc}, each I-rule for conjunction gives rise to two distinct E-rules.
I$\times$1 gives rise to E$\times$1a and E$\times$1b, while I$\times$2 gives rise to E$\times$2a and E$\times$2b.
Each E-rule derives a premise of the corresponding I-rule from its conclusion together with its other premise.
In this way, we can obtain an implication from an atom, or vice-versa.
The proviso that the probability $g$ of the minor premise be different from $0$ is added to make $f/g$ well defined. 
The intended meaning of this restriction is the following.
The minor premise of any E$\times$ represents $Pr(B | A)$ (or $Pr(A)$), and if it has probability $0$, also $Pr(A\times B)$ has probability $0$.
But then, any probability of the other conjunct $Pr(A)$ (or $Pr(B | A)$) is consistent with $Pr(A\times B) = 0$.
Hence, we cannot derive any specific probability for such judgments.\footnote{
Note that if $g = 0$, then $f=0$.
However, we do not need to impose this ``consonance'' in the E-rules for $\times$, since it is automatically satisfied by every derivation of their premises.
}

\begin{example}[Conjunction]
As an example, from the judgment seen in example~\ref{exa:Cond}, together with the atomic judgment that ascribes the female gender to the subject described by:
\medskip
\[
\sigma:= Age: 27 ,\: M.S.: married + divorced,\: Etn.: white^{\bot} \rhd  Gen. : f _{0.50}
\]

\medskip
\noindent
we can obtain the probability of the conjunction ascribing both gender and success in obtaining the loan:
\medskip
\[
Age: 27 ,\: M.S.: married + divorced,\: Etn.: white^{\bot} \rhd  \langle Gen. , Loan \rangle : f \times yes _{0.30}
\]

\end{example}

Let us now consider the rules in table~\ref{tab:LRtriangle}, starting with those for $ + $.
The rule I$\rhd+$ introduces the disjunction on the right, relying on the property defining the probability of mutually exclusive alternatives:
\medskip
\[
Pr(A+B) = Pr (A) + Pr (B)
\]

\medskip
\noindent
Note that premise $ u: \delta ~ \Bot ~ u: \gamma  $ formalizes the requirement that $\delta$ and $\gamma$ are mutually exclusive values for $u$.
Moreover, note that, for the atomic values of atomic variables, this requirement is automatically satisfied.
For logically complex variables and values, the formal criteria for establishing mutual exclusivity are addressed in~\ref{appendix}.\footnote{
Since every value correspond to at most one variable, sometimes we will omit the variables and write directly $ \delta ~ \Bot ~ \gamma  $.
}

The corresponding E-rules E$\rhd+$a and E$\rhd+$b use the probability of the attribution of a disjunction of values and the probability of the attribution of one such value to derive the probability of the attribution of the other value.
Also in this case, the mutual exclusivity of the disjuncts is needed for the rule to be sound.
However, since the derivability of the disjunction in the major premise already requires such mutual exclusivity, we can drop this assumption from the premises.

The rule I$+\rhd$ introduces the disjunction on the left of $\rhd$, that is it introduces the disjunction in the list of values attributions that characterize the subject of the query.
In general, the justification of the left-rules is more complex, because they operate \textit{on the condition} of a conditional probability.
In this case, the justification relies on the following provable property of probability:
\medskip
\[
Pr(C| A+B) = \dfrac{Pr(C| A)\cdot Pr(A)+Pr(C| B)\cdot Pr(B)}{Pr(A) + Pr(B)}
\]

\medskip
\noindent
which can be proved as follows: 

\bigskip

\scalebox{0.80}{$
   \begin{aligned}
Pr(C| A+B) &=^{def}  \dfrac{Pr((A+B)\times C)}{Pr(A+B)} =^{distri. \times on +} 
\dfrac{Pr((A \times C)+(B\times C))}{Pr(A+B)} =  \\
^{Pr((A \times C)\times(B\times C))=0} &= \dfrac{Pr(A\times C)+Pr(B\times C)}{Pr(A+B)} = ^{def}
\dfrac{Pr(C| A)\cdot Pr(A) + Pr(C| B)\cdot Pr(B)}{Pr(A+B)} = \\
^{Pr(A\times B)=0} &= \dfrac{Pr(C| A)\cdot Pr(A)+Pr(C| B)\cdot Pr(B)}{Pr(A) + Pr(B)}
\end{aligned}$}

\bigskip

Using this property, we can obtain the probability of the value attribution $u: \delta $ under the condition $ \sigma , t: (\gamma + \beta) $, using:
\begin{itemize}
\item The probability of the value attribution $u: \delta $ under the condition $ \sigma , t: \gamma $;
\item The probability of the value attribution $u: \delta $ under the condition $ \sigma , t: \beta $;
\item The probability of the value attribution $t: \gamma $ under the condition $ \sigma $;
\item The probability of the value attribution $t: \beta $ under the condition $ \sigma $.
\end{itemize}

\noindent
I$+\rhd$ takes all these judgments as premises and derives its conclusion from them.
Like for I$\rhd +$, also in this case, we need the extra premise $ t: \gamma ~ \Bot ~ t: \beta  $, which formalizes the requirement that the values composing the disjunction are mutually exclusive.

Given definition~\ref{def:InvPrinc}, there are four E-rules corresponding to I$+\rhd$, one for each of its premises.
They can be used both to eliminate the disjunction from the condition (E$+\rhd$a and E$+\rhd$b) and to derive the probability of the disjuncts (E$+\rhd$c and E$+\rhd$d).

\begin{example}[Disjunction]
As an example, the judgment seen in example~\ref{exa:Cond} has the disjunction ``married or divorced'' in $ \sigma $, as the value of the variable ``marital status''.
We can use E$+\rhd$a to infer the probability of receiving a loan for a married woman, if we assume the following judgments:
\medskip
\[
Age: 27 ,\: M.S.: divorced,\: Etn.: white^{\bot} ,\: Gen. : f \rhd Loan : yes _{0.40}
\]
\[
Age: 27 ,\: Etn.: white^{\bot} ,\: Gen.: f \rhd  M.S.: divorced _{0.10}
\]
\[
Age: 27 ,\: Etn.: white^{\bot} ,\: Gen.: f \rhd  M.S.: married _{0.45}
\]

\bigskip
Indeed, by applying the rule and approximating $\frac{0.60\cdot (0.45+0.10)-0.40\cdot 0.10}{0.45} $ to $ 0.64$, we find:
\medskip
\[
Age: 27 ,\: M.S.: married, \: Etn.: white^{\bot} ,\: Gen. : f \rhd Loan : yes _{\approx 0.64}
\]

\medskip
\noindent
Hence, we can conclude that the probability of receiving a loan for a married woman is $0.64$, much higher than that for a divorced woman ($0.40$).
\end{example}

\begin{table} %è una tabella floating, ma va comunuqe scritta approssimativamente dove si vuole che compaia.
\caption{Left and Right Logical Rules} %didascalia
\label{tab:LRtriangle} 
\centering %meglio di center, per la spaziatura
\scalebox{.85}{
\begin{tabular}{c}

\AxiomC{$ \sigma \rhd u: \delta _{f}  $}
	\AxiomC{$ \sigma \rhd u: \gamma _{g}  $}
    \AxiomC{$ u: \delta ~ \Bot ~ u: \gamma  $}
	\RightLabel{{\tiny I$\rhd+$}}
	\TrinaryInfC{$ \sigma \rhd u: (\delta + \gamma) _{f+g}  $}
\DisplayProof

\\[0.8cm]

\AxiomC{$ \sigma \rhd u: (\delta + \gamma) _{h} $}
	\AxiomC{$ \sigma \rhd u: \delta _{f}  $}
	\RightLabel{{\tiny E$\rhd+$a}}
	\BinaryInfC{$ \sigma \rhd u: \gamma _{h-f}  $}
\DisplayProof

\\[0.8cm]

\AxiomC{$ \sigma \rhd u: (\delta + \gamma) _{h} $}
	\AxiomC{$ \sigma \rhd u: \gamma _{g}  $}
	\RightLabel{{\tiny E$\rhd+$b}}
	\BinaryInfC{$ \sigma \rhd u: \delta _{h-g}  $}
\DisplayProof

\\[0.8cm]

\AxiomC{$ \sigma , t: \gamma \rhd u: \delta _{f}  $}
	\AxiomC{$ \sigma , t: \beta \rhd u: \delta _{g}  $}
	\AxiomC{$ \sigma \rhd t: \gamma _{h}  $}
	\AxiomC{$ \sigma \rhd t: \beta _{i}  $}
    \AxiomC{$ t: \gamma ~ \Bot ~ t: \beta  $}
	\RightLabel{{\tiny I$+\rhd$}}
	\QuinaryInfC{$ \sigma , t: (\gamma + \beta) \rhd u: \delta _{\frac{f\cdot h + g \cdot i}{h+i}}  $}
\DisplayProof

\\[0.8cm]

\AxiomC{$ \sigma , t: \gamma + \beta \rhd u: \delta _{f}  $}
	\AxiomC{$ \sigma , t: \beta \rhd u: \delta _{g}  $}
	\AxiomC{$ \sigma \rhd t: \gamma _{h}  $}
	\AxiomC{$ \sigma \rhd t: \beta _{i}  $}
	\RightLabel{{\tiny E$+\rhd$a}}
	\QuaternaryInfC{$ \sigma , t: \gamma \rhd u: \delta _{\frac{f\cdot (h+i) - g\cdot i}{h}}  $}
\DisplayProof

\\[0.8cm]

\AxiomC{$ \sigma , t: \gamma + \beta \rhd u: \delta _{f}  $}
	\AxiomC{$ \sigma , t: \gamma \rhd u: \delta _{g}  $}
	\AxiomC{$ \sigma \rhd t: \gamma _{h}  $}
	\AxiomC{$ \sigma \rhd t: \beta _{i}  $}
	\RightLabel{{\tiny E$+\rhd$b}}
	\QuaternaryInfC{$ \sigma , t: \beta \rhd u: \delta _{\frac{f\cdot (h + i)-g\cdot  h}{i}}  $}
\DisplayProof

\\[0.8cm]

\AxiomC{$ \sigma , t: \gamma \rhd u: \delta _{f}  $}
\AxiomC{$ \sigma , t: \beta \rhd u: \delta _{g}  $}
\AxiomC{$ \sigma , t: \gamma + \beta \rhd u: \delta _{h}  $}
	\AxiomC{$ \sigma \rhd t: \beta _{i}  $}
	\RightLabel{{\tiny E$+\rhd$c}}
	\QuaternaryInfC{$ \sigma \rhd t: \gamma _{\frac{i\cdot (g-h)}{h - f}}  $}
\DisplayProof

\\[0.8cm]

\AxiomC{$ \sigma , t: \gamma \rhd u: \delta _{f}  $}
\AxiomC{$ \sigma , t: \beta \rhd u: \delta _{g}  $}
\AxiomC{$ \sigma , t: \gamma + \beta \rhd u: \delta _{h}  $}
	\AxiomC{$ \sigma \rhd t: \gamma _{i}  $}
	\RightLabel{{\tiny E$+\rhd$d}}
	\QuaternaryInfC{$ \sigma \rhd t: \beta _{\frac{i\cdot (f- h)}{h - g}}  $}
\DisplayProof

\\[0.8cm]

\AxiomC{$ \sigma \rhd u: \delta _{g} $}
	\RightLabel{{\tiny I/E$\rhd\bot$}}
	\doubleLine
	\UnaryInfC{$ \sigma  \rhd u: \delta ^{\bot} _{1-g} $}
\DisplayProof

\\[0.8cm]

\def\defaultHypSeparation{\hskip .1in}
\AxiomC{$ \sigma \rhd t: \beta _{f} $}
\AxiomC{$ \sigma \rhd u: \delta _{g} $}
\AxiomC{$ \sigma , t: \beta \rhd u: \delta _{h} $}
	\RightLabel{{\tiny I$\bot\rhd$}}
	\TrinaryInfC{$ \sigma, t: \beta^{\bot} \rhd u: \delta _{\frac{g-f\cdot h}{1-f}} $}
\DisplayProof

\\[0.8cm]

\def\defaultHypSeparation{\hskip .1in}
\AxiomC{$ \sigma \rhd t: \beta _{f} $}
\AxiomC{$ \sigma \rhd u: \delta _{g} $}
\AxiomC{$ \sigma , t: \beta^{\bot} \rhd u: \delta _{h} $}
	\RightLabel{{\tiny E$\bot\rhd$a}}
	\TrinaryInfC{$ \sigma , t: \beta \rhd u: \delta _{\frac{g+ h \cdot(f-1)}{f}}$}
\DisplayProof

\\[0.8cm]

\def\defaultHypSeparation{\hskip .1in}
\AxiomC{$ \sigma \rhd t: \beta _{f} $}
\AxiomC{$ \sigma , t :\beta \rhd  u: \delta _{g} $}
\AxiomC{$ \sigma , t: \beta^{\bot} \rhd u: \delta _{h} $}
	\RightLabel{{\tiny E$\bot\rhd$b}}
	\TrinaryInfC{$ \sigma \rhd u: \delta _{h-h\cdot f + g \cdot f}$}
\DisplayProof

\\[0.8cm]

\def\defaultHypSeparation{\hskip .1in}
\AxiomC{$ \sigma \rhd u: \delta _{f} $}
\AxiomC{$ \sigma , t: \beta \rhd u: \delta _{g} $}
\AxiomC{$ \sigma , t: \beta^{\bot} \rhd u: \delta _{h} $}
	\RightLabel{{\tiny E$\bot\rhd$c}}
	\TrinaryInfC{$ \sigma \rhd t: \beta _{\frac{f-h}{g-h}}$}
\DisplayProof

\\
\end{tabular} }
\end{table}

The rules $\rhd \bot $ deal with the probability that the target variable \textit{does not} receive a specific value, which can be easily calculated as $ 1 $ (representing certainty) minus the probability that the target variable \textit{does} receive the value.

The rules $\bot \rhd$ are a little more complex, dealing with negation on the left of $\rhd$, that is on the list of values assignments characterizing the subject of the query.
I$\bot \rhd$ is based on the following property of probability:\footnote{
This property easily follows from the law of total probabilities ($Pr(B) = Pr(A\times B)+ Pr(A ^{\bot} \times B)$) and the probability of the negation ($Pr(A ^{\bot}) = 1 - Pr(A)$).
}
\medskip
\[
Pr(B| A ^{\bot}) = \dfrac{Pr(B) - Pr(B| A)\cdot Pr(A)}{1 - Pr(A)}
\]

\medskip
Given definition~\ref{def:InvPrinc}, there are three E-rules corresponding to I$\bot \rhd$, one for each of its premises.
They can be used to eliminate the negation from the condition (E$\bot \rhd$a) and to derive the probability of the negated formula or of the target variable receiving the value when the condition is dropped (E$\bot \rhd$b and E$\bot \rhd$c).

\begin{example}[Negation]
As an example, the judgment seen in example~\ref{exa:Cond} has the negation ``not white'' in $ \sigma $, as the value of the variable ``ethnicity''.
We can use E$\bot \rhd$a to infer the probability of receiving a loan for a white woman, if we assume the following judgments:
\medskip
\[
Age: 27 ,\: M.S.: married + divorced,\: Gen. : f \rhd Loan : yes _{0.75}
\]
\[
Age: 27 ,\: M.S.: married + divorced,\: Gen. : f \rhd Etn.: white _{0.80}
\]

\medskip
Indeed, by applying the rule and approximating $\frac{0.75 + 0.60 \cdot (0.80 - 1)}{0.80}$ to $0.79$, we find:
\medskip
\[
Age: 27 ,\: M.S.: married + divorced, \: Etn.: white,\: Gen. : f \rhd Loan : yes _{\approx 0.79}
\]

\medskip
\noindent
Hence, we can conclude that the probability of receiving a loan for a white woman is $0.79$, much higher than those for a non-white woman ($0.60$).
\end{example}

\subsubsection{Extension of the language}
\label{subsub:ExtLang}

To make the system as simple as possible, we have kept the language very minimal.
However, in order to investigate the notion of copy of an ML system, it is convenient to have a system in which Learning Algorithm and Training Set are explicit in the notation.
Hence, we use the following notation
\medskip
\[
\Gamma \vdash A_{\sigma}(t): \beta  _{f}
\]

\medskip
\noindent
to mean that an ML system that implements Learning Algorithm $ A $ and is trained using Training Set $ \Gamma $ satisfies $ \sigma  \rhd t: \beta  _{f} $.
All the rules remain valid when we substitute $ \sigma  \rhd t: \beta  _{f} $ with $\Gamma \vdash A_{\sigma}(t): \beta  _{f}$, provided that we require the same Training Set $ \Gamma $ and Learning Algorithm $A$ to occur in all the premises and in the conclusion.
That is, logical rules compose and decompose logically variables and values \textit{for a given ML system}.
As an example, I$\times$1 becomes:

\begin{prooftree}
\AxiomC{$\Gamma \vdash A_{\sigma ,t: \beta}(u) : \delta _{f} $}
	\AxiomC{$ \Gamma \vdash A_{\sigma}(t): \beta _{g} $}
	\RightLabel{{\tiny I$\times$1}}
	\BinaryInfC{$ \Gamma \vdash A_{\sigma}(\langle t, u \rangle ): (\beta\times  \delta) _{f\cdot g} $}
\end{prooftree}

As for atomic judgments, we will use:
\medskip
\[
\Gamma \vdash A_{\sigma}(a): \alpha ^{1} _{f_{1}} \ldots \alpha ^{n} _{f_{n}}
\]

\medskip
\noindent
to mean that $\Gamma \vdash A_{\sigma}(a): \alpha ^{i} _{f_{i}}$ for every $ i $ s.t. $1 \leq i \leq n$ and $ \alpha ^{1} \ldots \alpha ^{n} $ are all the possible atomic values of $ a $.

\section{Trustworthy Copies}
\label{sec:TrustCopies}

In this section, we start to investigate the notion of trustworthiness of a copy of an ML system with respect to its original.
We begin by categorizing different kinds of copy of an original ML system and explaining how they are distinguished in our notation.
Then, we investigate different notions of trustworthiness for copies: Justifiably, Weakly, Equally and Almost Trustworthiness.

Different ways of defining trustworthiness of an ML system with respect to its original may depend on whether one focuses on local or global trustworthiness, i.e. on whether we evaluate single applications of the system to specific queries or trustworthiness across different applications.
Moreover, one may focus on queries with a specific logical structures.
In the present and the following section~\ref{sec:RelationsTrust}, we will be local and atomic: we will address trustworthiness only concerning \textit{single applications} of ML systems to specific lists of values attributions $\sigma$ and for specific atomic target variables $ a $.
Moreover, we will consider only atomic values $\alpha ^{1}, \ldots , \alpha ^{n}$ (assumed to be both exhaustive and mutually exclusive) for these target variables.
In summary, we will consider only trustworthiness with respect to single queries corresponding to atomic judgments.

This restricted notion of trustworthiness  will constitute the main core of our investigation.
We will retain these restrictions up until Section~\ref{sec:TrustLogic}, where we will investigate trustworthiness across different applications of an ML system and for queries that possibly regard logically complex variables and values.
We will see that for both these liberalizations there are natural extensions of the definitions seen in this section. 

\subsection{Different notions of copy}
\label{subsec:DifferentCopies}

Since there are two components in an ML system (Learning Algorithm and Training Set), we can distinguish two different kinds of copy of an original system, depending on the component that varies in the copy.
Of course, ML systems could diverge in both their components, but we will change only one of them at a time, in order to investigate the consequences of such a change \textit{ceteris paribus}.

\subsubsection{Different Learning Algorithm}

Let us first of all consider systems that implement different Learning Algorithms, while their Training Set and Test Set remain unchanged and so does obviously the intended task.
Notationally, this distinction is displayed by a variation of the uppercase Latin letter for the Learning Algorithm.
As an example, we can compare the probabilities of the atomic values $\alpha ^{1} \ldots \alpha ^{n}$ for the atomic variable $ a $ when two ML systems, one that implements Learning Algorithm $ A $ and another that implements Learning Algorithm $ B $, are trained using the same Training Set $\Gamma$ and interrogated using the same list of values attributions $ \sigma $:
\[
\Gamma \vdash A_{\sigma}(a): \alpha ^{1} _{f_{1}} \ldots \alpha ^{n} _{f_{n}}
\]
\[
\Gamma \vdash B_{\sigma}(a): \alpha ^{1} _{g_{1}} \ldots \alpha ^{n} _{g_{n}}
\]

\begin{example}[Different Learning Algorithms]
\label{exa:diffLA}
Let us consider two systems devised to decide the appropriate sentence for an offender and assume that they run different Learning Algorithms but are trained with the same Training Set of defendants convicted with various sentences.
The Algorithms could be different for several reasons.
As an example, only one of them could implement the following protection:
\begin{description}
\item[Fairness through unawareness:] the value of the variable ``ethnicity'' should not be used for deciding the sentence.
\end{description}
As a result of this difference, the two systems could behave differently when applied to the same list of values $ \sigma $ from the Test Set (the characteristics of a person found guilty for which the system has to decide the sentence).
\end{example}

\subsubsection{Different Training Set}

Let us now consider systems that implement the same Learning Algorithm, but are trained using different Training Sets.
Notationally, this distinction is displayed by a variation of the uppercase Greek letter for the Training Set.
As an example, we can compare the probabilities of the atomic values $\alpha ^{1} \ldots \alpha ^{n}$ for the atomic variable $ a $ when two ML systems, one trained with Training Set $\Gamma$ and another trained with Training Set $\Delta$, implement the same Learning Algorithm $ A $ and are interrogated using the same list of values attributions $ \sigma $:
\medskip
\[
\Gamma \vdash A_{\sigma}(a): \alpha ^{1} _{f_{1}} \ldots \alpha ^{n} _{f_{n}}
\]
\[
\Delta \vdash A_{\sigma}(a): \alpha ^{1} _{g_{1}} \ldots \alpha ^{n} _{g_{n}}
\]

\begin{example}[Different Training Set]

Let us consider again a pair of systems devised to decide the appropriate sentence for an offender.
Moreover, let us assume that they implement the same Learning Algorithm, but we do not know whether these systems satisfy \textit{fairness through unawareness}.
Then, we decide to train these systems differently:
\begin{itemize}
\item The first system is trained using a fair Training Set, that is one carefully selected so that, if two offenders share all the properties apart from race, they receive the same sentence;
\item The second system is trained using a biased Training Set, that is one in which there are cases of offenders sharing all properties apart from race but receiving different sentences.
\end{itemize}
When we test these systems, if we observe that only the first one behaves fairly, then we can conclude that the Leaning Algorithm has no \textit{built-in} protection, and fairness is achieved only during training.

\end{example}

\subsection{Different notions of trust}
\label{subsec:DiffTrust}

In this section, we define different notions of trustworthiness for a copy with respect to an original system.
As we have seen in subsection~\ref{subsec:DifferentCopies}, copies can diverge from their original system for different reasons: they can be based on different Learning Algorithms or they can be trained using different Training Sets.
Our notions of trustworthiness apply to both kinds of copy.
However, the epistemic gain obtained with trustworthy copies depends on the specific kind of copy.
Indeed, our trustworthiness criteria also evaluate the adequacy of the different Training Sets and Learning Algorithms.

\subsubsection{Justifiably Trustworthy}

A first, quite obvious, kind of trustworthiness regards systems that share the same behavior for all the possible values of a variable.
In this case, the systems attach the same probabilities to all these values.
Although this is a case in which trust is evident, it is still relevant for two reasons:
\begin{itemize}
\item It will work as a limit case of other, weaker, notions of trust;
\item It is useful to interpret the consequences of trustworthiness for the evaluation of Training Sets and Learning Algorithms.
\end{itemize}
Note that, since we deal only with atomic variables and atomic values, we can easily define this kind of trustworthiness by taking into account all the possible values $\alpha^{1}, \dots ,\alpha^{n}$ for the atomic variable at hand $ a $, since these are finite in number.

\subparagraph*{Different Learning Algorithms}

Let us first of all compare systems that share the same Training Set, but have different Learning Algorithms.
Of course, a copy is trustworthy with respect to its original, if all of their outputs coincide.

\begin{definition}[Justifiably Trustworthy copies with different Learning Algorithms]
Given two ML systems dubbed respectively the \textit{original} and the \textit{copy}, composed of different Learning Algorithms $A$ and $B$ and trained with the same Data Set $\Gamma$, we say that the \textit{copy} is Justifiably Trustworthy with respect to the \textit{original} considering feature $ a $ and the same row of values $ \sigma $ from the Test Set as the original iff $ \Gamma \vdash A_{\sigma}(a): \alpha ^{1} _{f_{1}} \ldots \alpha ^{n} _{f_{n}} $ and $ \Gamma \vdash B_{\sigma}(a): \alpha ^{1} _{g_{1}} \ldots \alpha ^{n} _{g_{n}} $, where $ \forall \alpha^{1}\dots\alpha^{n}\mid f_{i}=g_{i} $.
\end{definition}

Notice that, in this case, the relation of trustworthiness is symmetric; this will not be the case for other notions of trustworthiness.
In this case, the conclusion is not only that the copy is trustworthy with respect to the original system, but also that the Learning Algorithm $B$ behaves as reliably as $A$, when both are trained with $ \Gamma $ and applied to the input $\sigma$, even though they are (possibly) designed in different ways.

\begin{example}[JT with different L.A.: credit risk evaluation]
\label{exa:JTdiffLA}
Let us consider two systems for the evaluation of credit risk based on different Learning Algorithms, but trained and applied to the same Data Sets.
Let us moreover assume that we trust one of these systems (the original), since it seems to avoid biases when used to predict credit risk, but we do not know whether this is because of some protection implemented in the Learning Algorithm, or because of a fair Training Set.
If we know that the copy behaves like the original, of course we know that it can avoid biases as well. Hence, since it has been trained using a different Learning Algorithm, we have reasons to believe that fairness is not achieved through some protections in the Learning Algorithm: not only the copy itself is trustworthy in this specific application, but also the Training Set in itself seems trustworthy when applied for purposes analogous to the ones considered here.
Of course, by repeated changes in the Learning Algorithm and applications to new rows of the Test Set, we can confirm both the trustworthiness of the copy and the unbiased nature of the Training Set.
\end{example}

\subparagraph*{Different Training Set}

Let us now compare systems that share the same Learning Algorithm and Test Set, but have been trained using different Training Sets.

\begin{definition}[Justifiably Trustworthy copies with different Training Sets]
Given two ML systems dubbed respectively the \textit{original} and the \textit{copy}, composed of the same Learning Algorithms $A$ but trained with different Data Sets $\Gamma$ and $\Delta$, we say that the \textit{copy} is Justifiably Trustworthy with respect to the \textit{original} considering feature $ a $ and the same row of values $ \sigma $ from the same Test Set as the original iff $ \Gamma \vdash A_{\sigma}(a): \alpha ^{1} _{f_{1}} \ldots \alpha ^{n} _{f_{n}} $ and $ \Delta \vdash A_{\sigma}(a): \alpha ^{1} _{g_{1}} \ldots \alpha ^{n} _{g_{n}} $, where $ \forall \alpha^{1}\dots\alpha^{n}\mid f_{i}=g_{i} $.
\end{definition}

In this case, the conclusion is not only that the copy is trustworthy with respect to the original system, but also that the Learning Algorithm leads the ML system to find the same output for the input $\sigma$, regardless (to some extent) of the Training Set used to train the machine.
Of course, this might suggest the possibility that the Training Sets $\Gamma$ and $\Delta$, although superficially different, are deeply analogous to each other. In order to minimize this possibility, we have to test the system with multiple Training Sets, so as to avoid lucky coincidences.

\begin{example}[JT with different Train.Set.: diagnostic software]
\label{exa:JTdiffTrS}
Consider two diagnostic systems that implement the same Learning Algorithm, but that are trained with different Data Sets: the population of the patients of two different hospitals.
Let us moreover assume that we trust the original system when it is applied to the Test Set under analysis.
Since the copy is Justifiably Trustworthy with respect to the original, we know that we can trust it too.
Hence, we can wonder whether the populations of the two hospitals, although maybe superficially different, highlight analogous correlations between the relevant properties, or whether the Learning Algorithm is good at dismissing spurious correlations and biases from the Training Sets.
Of course, by repeated changes in the Training Sets and applications of these systems to new rows of the Test Set we can confirm both the trustworthiness of the copy and the unbiased nature of the Learning Algorithm.
\end{example}

\subsubsection{Equally Trustworthy}

The previous notion of Justifiably Trustworthiness is clearly too stringent: it asks that copy and original behave exactly the same for all the possible values of the variable. A less stringent notion of trustworthiness requires that they behave in the same way only regarding a \textit{relevant} subset of the possible values of the variable. We will say that the copy is Equally Trustworthy regarding the subset of \textit{relevant} values if and only if the original and the copy attribute the same probabilities to these values.
Since the order of the possible outputs $ \alpha ^{1}, ..., \alpha ^{n} $ is arbitrary, we will assume that the relevant outputs are always the first values $ \alpha ^{1}, ..., \alpha ^{m} $ with $ m \leq n $.

To see why such a notion of trustworthiness is useful, let us consider the following example.

\begin{example}[ET: diagnostic software]
\label{exa:ET}

Let us consider a diagnostic software for the detection of Hepatitis, and assume that it assigns to a patient the probability that they contracted this disease, with the following possible outputs:  $absent$, $minor$, $moderate$, $major$, and $extreme$.
Moreover, let us assume that we want to compare this system with another diagnostic software for the detection of Hepatitis and that we are interested only in moderate to extreme cases.
Of course, we can ignore all the probabilities attached to the values $absent$ and $minor$ by these systems, focusing only on the relevant outputs: $moderate$, $major$, and $extreme$.
Hence, we can ask that copy and original behave accordingly only regarding these values, that is that they are Equally Trustworthy regarding these values.
\end{example}

Note that, like for Justifiably Trustworthiness, also in this case the relation is symmetric.
Moreover, like for Justifiably Trustworthy copies, also for Equally Trustworthy copies we can distinguish different cases, depending on which component of the system is changed in the copy.

\subparagraph*{Different Learning Algorithms}

\begin{definition}[Equally Trustworthy copies with different Learning Algorithms]
Given two ML systems dubbed respectively the \textit{original} and the \textit{copy}, composed of different Learning Algorithms $A$ and $B$ and trained with the same Data Set $\Gamma$, we say that the \textit{copy} is Equally Trustworthy with respect to the \textit{original} regarding values $ \alpha ^{1} \ldots \alpha ^{m} $ (with $ m \leq n $) when considering feature $ a $ and the same row of values $ \sigma $ from the same Test Set as the original iff $ \Gamma \vdash A_{\sigma}(a): \alpha ^{1} _{f_{1}} \ldots \alpha ^{n} _{f_{n}} $ and $ \Gamma \vdash B_{\sigma}(a): \alpha ^{1} _{g_{1}} \ldots \alpha ^{n} _{g_{n}} $, where $ \forall i \: _{1\leq i \leq m} \; \mid g_{i}= f_{i}$.
\end{definition}

\subparagraph*{Different Training Set}

\begin{definition}[Equally Trustworthy copies with different Training Sets]
Given two ML systems dubbed respectively the \textit{original} and the \textit{copy}, composed of the same Learning Algorithm $A$ but trained with different Data Sets $\Gamma$ and $\Delta$, we say that the \textit{copy} is Equally Trustworthy with respect to the \textit{original} regarding values $ \alpha ^{1} \ldots \alpha ^{m} $ (with $ m \leq n $) when considering feature $ a $ and the same row of values $ \sigma $ from the same Test Set as the original iff $ \Gamma \vdash A_{\sigma}(a): \alpha ^{1} _{f_{1}} \ldots \alpha ^{n} _{f_{n}} $ and $ \Delta \vdash A_{\sigma}(a): \alpha ^{1} _{g_{1}} \ldots \alpha ^{n} _{g_{n}} $, where $ \forall i \: _{1\leq i \leq m} \; \mid g_{i}= f_{i}$.
\end{definition}

\subsubsection{Weakly Trustworthy}

Restricting the set of values for which copy and original are required to behave equivalently is not the only way of relaxing the notion of trustworthiness.
Another quite flexible notion of trustworthiness requires that the copy produces the \textit{relevant} values with \textit{equal or higher} frequency.
This essentially corresponds to the requirement that the frequency of the \textit{relevant} values does not decrease, i.e. to an application of a precautionary principle.
Note, moreover, that when we want to impose the opposite requirement that the \textit{relevant} values do not increase in frequency, we just have to select the complement of the \textit{relevant} values.
Hence, the definition is more flexible than it could seem at first glance.

Another restriction that we will impose here (we will relax it in the next subsection) is that the set of values for the output should be the same for the original system and the copy.
That is, even though the frequency of the output can be different, both systems should consider the same set of possible outputs for the variable.
Hence, we will call a copy Weakly Trustworthy regarding the subset of \textit{relevant} values iff the copy attributes the same or higher probabilities to these values than the original, and moreover if copy and original attribute to the variable the same set of possible values.
Note that this relation is not symmetric, as opposed to Justifiably and Equally Trustworthiness.

To see why such a notion of trustworthiness is useful, let us consider the following extension of example~\ref{exa:ET}.
\begin{example}[WT: diagnostic software]
\label{exa:WT}
Let us consider a pair of diagnostic systems, labeled the \textit{original} and the \textit{copy}, devised for the detection of \textit{Hepatitis}, and such that they give the probabilities of its severity, with the following possible outputs: $absent$, $minor$, $moderate$, $major$, and $extreme$.
Moreover, let us assume that the copy lightly overestimates the probability of $moderate$, $major$, and $extreme$ cases of \textit{Hepatitis}.
We could decide, based on pragmatic reasons, that such a system is trustworthy enough, especially when it is used to suggest further tests to patients.
\end{example}

Let us see the definition of Weak Trustworthiness for the different kinds of copy.

\subparagraph*{Different Learning Algorithms}

\begin{definition}[Weakly Trustworthy copies with different Learning Algorithms]
Given two ML systems dubbed respectively the \textit{original} and the \textit{copy}, composed of different Learning Algorithms $A$ and $B$ and trained with the same Data Set $\Gamma$, we say that the \textit{copy} is Weakly Trustworthy with respect to the \textit{original} regarding values $ \alpha ^{1} \ldots \alpha ^{m} $ (with $ m \leq n $) when considering feature $ a $ and the same row of values $ \sigma $ from the same Test Set as the original iff $ \Gamma \vdash A_{\sigma}(a): \alpha ^{1} _{f_{1}} \ldots \alpha ^{n} _{f_{n}} $ and $ \Gamma \vdash B_{\sigma}(a): \alpha ^{1} _{g_{1}} \ldots \alpha ^{n} _{g_{n}} $, where $ \forall i \: _{1\leq i \leq m} \; \mid g_{i}\geq f_{i}$ and $\forall i \; \mid g_{i} \neq 0 \leftrightarrow f_{i} \neq 0$.
\end{definition}

\subparagraph*{Different Training Set}

\begin{definition}[Weakly Trustworthy copies with different Training Sets]
Given two ML systems dubbed respectively the \textit{original} and the \textit{copy}, composed of the same Learning Algorithm $A$ but trained with different Data Sets $\Gamma$ and $\Delta$, we say that the \textit{copy} is Weakly Trustworthy with respect to the \textit{original} regarding values $ \alpha ^{1} \ldots \alpha ^{m} $ (with $ m \leq n $) when considering feature $ a $ and the same row of values $ \sigma $ from the same Test Set as the original iff $ \Gamma \vdash A_{\sigma}(a): \alpha ^{1} _{f_{1}} \ldots \alpha ^{n} _{f_{n}} $ and $ \Delta \vdash A_{\sigma}(a): \alpha ^{1} _{g_{1}} \ldots \alpha ^{n} _{g_{n}} $, where $ \forall i \: _{1\leq i \leq m} \; \mid g_{i}\geq f_{i}$ and $\forall i \; \mid g_{i} \neq 0 \leftrightarrow f_{i} \neq 0$.
\end{definition}

\subsubsection{Almost Trustworthy}

The last and least stringent notion of trustworthiness that we will consider is Almost Trustworthiness.
A copy is Almost Trustworthy compared to its original if produces the \textit{relevant} values with \textit{equal or higher} frequency, but the set of values for the output is possibly different for the original system and the copy.
Note that this relation is essentially identical to Weak Trustworthiness without the assumption that the two systems work on the same possible values for the target variable.

To see why such a notion of trustworthiness is useful, let us reconsider example~\ref{exa:ET} as follows: 

\begin{example}[AT: diagnostic software]
\label{exa:AT}
Let us consider a pair of diagnostic systems, labeled the \textit{original} and the \textit{copy}, devised for the detection of \textit{Hepatitis}, and such that they give the probabilities of its severity.
Let us also assume that the original system considers $absent$, $minor$, $moderate$, $major$, and $extreme$ as possible outputs, while the copy only considers $absent$, $moderate$, $major$, and $extreme$.\footnote{
Formally, we will model this situation with the copy always attributing probability $0$ to a minor case of Hepatitis.
}
Moreover, let us assume that the copy lightly overestimates the probability of $moderate$, $major$, and $extreme$ cases of \textit{Hepatitis}.
Note that this means that the copy moves the $minor$ cases of \textit{Hepatitis} to $moderate$, $major$, and $extreme$ cases, and not to $absent$ cases.
Hence, we could decide, based on pragmatic reasons, that such a system is trustworthy enough, especially when it is used to suggest further tests to patients.
\end{example}

\subparagraph*{Different Learning Algorithms}

\begin{definition}[Almost Trustworthy copies with different Learning Algorithms]
Given two ML systems dubbed respectively the \textit{original} and the \textit{copy}, composed of different Learning Algorithms $A$ and $B$ and trained with the same Data Set $\Gamma$, we say that the \textit{copy} is Almost Trustworthy with respect to the \textit{original} regarding values $ \alpha ^{1} \ldots \alpha ^{m} $ (with $ m \leq n $) when considering feature $ a $ and the same row of values $ \sigma $ from the same Test Set as the original iff $ \Gamma \vdash A_{\sigma}(a): \alpha ^{1} _{f_{1}} \ldots \alpha ^{n} _{f_{n}} $ and $ \Gamma \vdash B_{\sigma}(a): \alpha ^{1} _{g_{1}} \ldots \alpha ^{n} _{g_{n}} $, where $ \forall i \: _{1\leq i \leq m} \; \mid g_{i}\geq f_{i}$.
\end{definition}

\subparagraph*{Different Training Set}

\begin{definition}[Almost Trustworthy copies with different Training Sets]
Given two ML systems dubbed respectively the \textit{original} and the \textit{copy}, composed of the same Learning Algorithm $A$ but trained with different Data Sets $\Gamma$ and $\Delta$, we say that the \textit{copy} is Almost Trustworthy with respect to the \textit{original} regarding values $ \alpha ^{1} \ldots \alpha ^{m} $ (with $ m \leq n $) when considering feature $ a $ and the same row of values $ \sigma $ from the same Test Set as the original iff $ \Gamma \vdash A_{\sigma}(a): \alpha ^{1} _{f_{1}} \ldots \alpha ^{n} _{f_{n}} $ and $ \Delta \vdash A_{\sigma}(a): \alpha ^{1} _{g_{1}} \ldots \alpha ^{n} _{g_{n}} $, where $ \forall i \: _{1\leq i \leq m} \; \mid g_{i}\geq f_{i}$.
\end{definition}

\section{Relations between trustworthiness notions}
\label{sec:RelationsTrust}

Although we interpreted trustworthiness in different ways, depending on the kind of copy under investigation, the relation of entailment between them can be described in general, without taking into account the specific kind of copy.
Let us use the first lowercase letters of the gothic alphabet $ \textgoth{a},\textgoth{b},\textgoth{c},... $ (possibly indexed with natural numbers $ \textgoth{a}_{n},\textgoth{b}_{n},\textgoth{c}_{n},... $) to indicate applications of ML systems to a string $ \sigma $ of data and regarding a target variable $t$: $\langle \Gamma , A , \sigma , t\rangle$.
In this way, we will be able to speak of copies in general, without specifying whether the Training Set $\Gamma$ or the Learning Algorithm $A$ is changed in the copy.
Moreover, as in the previous section, let us consider only ML systems applied to atomic target variables $a$, and let us focus only on their possible atomic values $\alpha ^{1}, \ldots , \alpha ^{n}$.

\subsection{Entailment between trustworthiness notions}
\label{subsec:entailment}

\begin{table}
\caption{Fundamental properties of Trust Relations} %didascalia
\label{tab:propertiesTrust} 
\centering
\scalebox{0.85}{
\begin{tabular}{|c|c||c|c|}
\multicolumn{2}{c}{\textbf{Justifiably Trust}} & \multicolumn{2}{c}{\textbf{Equally Trust}}\\
\hline 
\textit{JT transitivity} & $\textgoth{a} \overset{JT}{\sim} \textgoth{b} \land \textgoth{b} \overset{JT}{\sim} \textgoth{c} \rightarrow \textgoth{a} \overset{JT}{\sim} \textgoth{c} $ & \textit{ET transitivity} & $\textgoth{a} \overset{ET}{{\underset{m}{\sim}}} \textgoth{b} \land \textgoth{b} \overset{ET}{{\underset{l}{\sim}}} \textgoth{c} \rightarrow \textgoth{a} \overset{ET}{{\underset{min(m,l)}{\sim}}} \textgoth{c} $ \\ 
\hline 
& & \textit{ET transitivity'} & $\textgoth{a} \overset{ET}{{\underset{m}{\sim}}} \textgoth{b} \land \textgoth{b} \overset{ET}{{\underset{m}{\sim}}} \textgoth{c} \rightarrow \textgoth{a} \overset{ET}{{\underset{m}{\sim}}} \textgoth{c} $\\ 
\hline 
\textit{JT reflexivity} & $\textgoth{a} \overset{JT}{\sim} \textgoth{a}$ & \textit{ET reflexivity} & $\textgoth{a} \overset{ET}{{\underset{m}{\sim}}} \textgoth{a}$\\ 
\hline 
\textit{JT symmetry} & $\textgoth{a} \overset{JT}{\sim} \textgoth{b} \rightarrow \textgoth{b} \overset{JT}{\sim} \textgoth{a} $ & \textit{ET symmetry} & $\textgoth{a} \overset{ET}{{\underset{m}{\sim}}} \textgoth{b} \rightarrow \textgoth{b} \overset{ET}{{\underset{l}{\sim}}} \textgoth{a} $, with $l \leq m$\\ 
\hline 
& & \textit{ET weakening} & $\textgoth{a} \overset{ET}{{\underset{m}{\sim}}} \textgoth{b} \rightarrow \textgoth{a} \overset{ET}{{\underset{l}{\sim}}} \textgoth{b} $, with $l \leq m$ \\ 
\hline
\multicolumn{4}{c}{}\\
\multicolumn{2}{c}{\textbf{Weakly Trust}} & \multicolumn{2}{c}{\textbf{Almost Trust}}\\
\hline 
\textit{WT transitivity} & $\textgoth{a} \overset{WT}{{\underset{m}{\leadsto}}} \textgoth{b} \land \textgoth{b} \overset{WT}{{\underset{l}{\leadsto}}} \textgoth{c} \rightarrow \textgoth{a} \overset{WT}{{\underset{min(m,l)}{\leadsto}}} \textgoth{c} $ & \textit{AT transitivity} & $\textgoth{a} \overset{AT}{{\underset{m}{\leadsto}}} \textgoth{b} \land \textgoth{b} \overset{AT}{{\underset{l}{\leadsto}}} \textgoth{c} \rightarrow \textgoth{a} \overset{AT}{{\underset{min(m,l)}{\leadsto}}} \textgoth{c} $\\ 
\hline 
\textit{WT transitivity'} & $\textgoth{a} \overset{WT}{{\underset{m}{\leadsto}}} \textgoth{b} \land \textgoth{b} \overset{WT}{{\underset{m}{\leadsto}}} \textgoth{c} \rightarrow \textgoth{a} \overset{WT}{{\underset{m}{\leadsto}}} \textgoth{c} $ & \textit{AT transitivity'} & $\textgoth{a} \overset{AT}{{\underset{m}{\leadsto}}} \textgoth{b} \land \textgoth{b} \overset{AT}{{\underset{m}{\leadsto}}} \textgoth{c} \rightarrow \textgoth{a} \overset{AT}{{\underset{m}{\leadsto}}} \textgoth{c} $\\ 
\hline 
\textit{WT reflexivity} & $\textgoth{a} \overset{WT}{{\underset{m}{\leadsto}}} \textgoth{a}$ & \textit{AT reflexivity} & $\textgoth{a} \overset{AT}{{\underset{m}{\leadsto}}} \textgoth{a}$ \\ 
\hline 
\textit{WT weakening} & $\textgoth{a} \overset{WT}{{\underset{m}{\leadsto}}} \textgoth{b} \rightarrow \textgoth{a} \overset{WT}{{\underset{l}{\leadsto}}} \textgoth{b} $, with $l \leq m$ & \textit{AT weakening} & $\textgoth{a} \overset{AT}{{\underset{m}{\leadsto}}} \textgoth{b} \rightarrow \textgoth{a} \overset{AT}{{\underset{l}{\leadsto}}} \textgoth{b} $, with $l \leq m$\\ 
\hline 
\end{tabular} 
}
\end{table}

First of all, with our new meta-variables, let us define some abbreviations for the Trust relations that we have defined in the previous section.
This will help us in comparing and composing them.

\begin{definition}[Trust Relations]
Abbreviations for the Trust relations:\

\begin{itemize}
\item $\textgoth{a} \overset{JT}{\sim} \textgoth{b} =_{df}$ $ \textgoth{a} $ is Justifiably Trustworthy with respect to $ \textgoth{b} $;
\item $\textgoth{a} \overset{ET}{\underset{m}{\sim}} \textgoth{b} =_{df}$ $ \textgoth{a} $ is Equally Trustworthy, with respect to $ \textgoth{b} $, regarding values $ \alpha ^{1}, \ldots , \alpha ^{m}$;
\item $\textgoth{a} \overset{WT}{\underset{m}{\leadsto}} \textgoth{b} =_{df}$ $ \textgoth{a} $ is Weakly Trustworthy, with respect to $ \textgoth{b} $, regarding values $ \alpha ^{1}, \ldots , \alpha ^{m}$;
\item $\textgoth{a} \overset{AT}{\underset{m}{\leadsto}} \textgoth{b} =_{df}$ $ \textgoth{a} $ is Almost Trustworthy, with respect to $ \textgoth{b} $, regarding values $ \alpha ^{1}, \ldots , \alpha ^{m}$.

\end{itemize}

\end{definition}

Let us now consider the fundamental properties of each trustworthiness relation.

\begin{observation}[Fundamental properties of Trust Relations]
JT, ET, WT and AT satisfy the following properties in table~\ref{tab:propertiesTrust}.
\end{observation}

\begin{proof}
The properties of JT follow from the same properties of $ = $.
The properties of ET follow from the same properties of $ = $ and obvious considerations concerning $ m $ and $ l $.
The properties of WT and AT follow from the same properties of $ \leq $ and obvious considerations concerning $ m $ and $ l $.
\end{proof}

\begin{table}
\caption{Coordination Principles} %didascalia
\label{tab:CoordinationP} 
\centering
\scalebox{0.90}{
\begin{tabular}{|c|c|}
\hline 
\textit{AT Bottom} & $(\textgoth{a} \overset{ET}{{\underset{m}{\leadsto}}} \textgoth{b} \lor \textgoth{a} \overset{WT}{{\underset{m}{\leadsto}}} \textgoth{b} \lor \textgoth{a} \overset{JT}{\sim} \textgoth{b})\rightarrow \textgoth{a} \overset{AT}{{\underset{m}{\leadsto}}} \textgoth{b} $ \\
\hline
\textit{JT Top} & $\textgoth{a} \overset{JT}{\sim} \textgoth{b} \rightarrow \forall m (\textgoth{a} \overset{AT}{{\underset{m}{\leadsto}}} \textgoth{b} \land \textgoth{a} \overset{ET}{{\underset{m}{\leadsto}}} \textgoth{b} \land \textgoth{a} \overset{WT}{{\underset{m}{\leadsto}}} \textgoth{b}) $ \\ 
\hline 
\textit{JT Top'} & $\textgoth{a} \overset{JT}{\sim} \textgoth{b} \rightarrow \forall m (\textgoth{a} \overset{ET}{{\underset{m}{\leadsto}}} \textgoth{b} \land \textgoth{a} \overset{WT}{{\underset{m}{\leadsto}}} \textgoth{b}) $ \\ 
\hline 
\textit{m=n to Top} & $((\textgoth{a} \overset{ET}{{\underset{m}{\leadsto}}} \textgoth{b} \lor \textgoth{a} \overset{WT}{{\underset{m}{\leadsto}}} \textgoth{b} \lor \textgoth{a} \overset{AT}{{\underset{m}{\leadsto}}} \textgoth{b})\land m=n) \rightarrow \textgoth{a} \overset{JT}{\sim} \textgoth{b} $ \\
\hline 
\textit{AT + m=n = Top} & $(\textgoth{a} \overset{AT}{{\underset{m}{\leadsto}}} \textgoth{b} \land m=n) \rightarrow \textgoth{a} \overset{JT}{\sim} \textgoth{b} $ \\
\hline
\textit{Semi-Antisymmetry AT} & $(\textgoth{a} \overset{AT}{{\underset{m}{\leadsto}}} \textgoth{b} \land \textgoth{b} \overset{AT}{{\underset{l}{\leadsto}}} \textgoth{a}) \rightarrow \textgoth{a} \overset{ET}{{\underset{min(m,l)}{\sim}}} \textgoth{b} $ \\ 
\hline 
\textit{Semi-Antisymmetry WT} & $(\textgoth{a} \overset{WT}{{\underset{m}{\leadsto}}} \textgoth{b} \land \textgoth{b} \overset{WT}{{\underset{l}{\leadsto}}} \textgoth{a}) \rightarrow \textgoth{a} \overset{ET}{{\underset{min(m,l)}{\sim}}} \textgoth{b} $ \\ 
\hline 
\end{tabular} 
}
\end{table}

Let us now see some \textit{coordination principles} that express the relations between \textit{different} notions of trustworthiness.

\begin{observation}[Coordination Principles]
JT, ET, WT, and AT are connected by the properties in table~\ref{tab:CoordinationP}.
\end{observation}

The basic relations of entailment between notions of Trust, which underlie all these coordination principles are summed up in figure~\ref{fig:entail}.
Before addressing composition of Trust relations, let us briefly discuss interdependence of these principles.

\begin{observation}[Derivability of Transitivity]
For ET, WT, and AT \textit{Transitivity} is derivable from \textit{Teakening} and \textit{Transitivity'}.
\end{observation}

\begin{proof}
The (simple) proof is as follows:
\begin{prooftree}
\AxiomC{$\textgoth{a} \overset{XT}{{\underset{m}{\leadsto}}} \textgoth{b}$}
\RightLabel{Weakening}
\UnaryInfC{$\textgoth{a} \overset{XT}{{\underset{min(m,l)}{\leadsto}}} \textgoth{b}$}

\AxiomC{$\textgoth{b} \overset{XT}{{\underset{l}{\leadsto}}} \textgoth{c}$}
\RightLabel{Weakening}
\UnaryInfC{$\textgoth{b} \overset{XT}{{\underset{min(m,l)}{\leadsto}}} \textgoth{c}$}

\RightLabel{Transitivity$^\prime$}
\BinaryInfC{$\textgoth{a} \overset{XT}{{\underset{min(m,l)}{\leadsto}}} \textgoth{c}$}
\end{prooftree}
\end{proof}

\begin{observation}[Derivability of \textit{m=n to Top}]
The property \textit{m=n to Top} is derivable from \textit{AT Bottom} and \textit{AT + m=n = Top}.
\end{observation}

\begin{proof}
The proof trivially follows for transitivity of implication.
\end{proof}

\begin{observation}[Derivability of \textit{JT Top}]
\textit{JT Top} can be derived from \textit{AT Bottom} and \textit{JT Top'}.
\begin{description}
\item[] $\textgoth{a} \overset{JT}{\sim} \textgoth{b} \rightarrow \forall m (\textgoth{a} \overset{ET}{{\underset{m}{\leadsto}}} \textgoth{b} \land \textgoth{a} \overset{WT}{{\underset{m}{\leadsto}}} \textgoth{b}) $.
\end{description}
\end{observation}

\begin{proof}
The proof trivially follows for transitivity of implication.
\end{proof}

\begin{observation}[No entailment between ET and WT]
None of the previous results establish an entailment relation between ET and WT.
In general: $\textgoth{a} \overset{ET}{{\underset{m}{\leadsto}}} \textgoth{b} \not\leftrightarrows \textgoth{a} \overset{WT}{{\underset{m}{\leadsto}}} \textgoth{b}$;
\end{observation}

\begin{figure}
\begin{subfigure}[h]{0.5\textwidth}
\scalebox{.55}{
\begin {tikzpicture}[-latex ,auto ,node distance =4 cm and 5cm ,on grid ,
semithick ,
state/.style ={ circle ,top color =white , bottom color = processblue!20 ,
draw,processblue , text=blue , minimum width =1 cm}]
\node[state] (C) {$JT$};
\node[state] (A) [above left=of C] {$ET$};
\node[state] (B) [above right =of C] {$WT$};
\node[state] (D) [above right =of A] {$AT$};
\path (C) edge [bend left =25] (A);
\path (C) edge [bend right =25] (B);
\path (A) edge [bend left =25] (D);
\path (B) edge [bend right =25] (D);
\path (C) edge (D);
\end{tikzpicture}
}
\caption{Entailment relations between notions of Trust}
\end{subfigure}
\begin{subfigure}[h]{0.5\textwidth}
\scalebox{.55}{
\begin {tikzpicture}[-latex ,auto ,node distance =4 cm and 5cm ,on grid ,
semithick ,
state/.style ={ circle ,top color =white , bottom color = processblue!20 ,
draw,processblue , text=blue , minimum width =1 cm}]
\node[state] (C) {$JT$};
\node[state] (A) [above left=of C] {$ET$};
\node[state] (B) [above right =of C] {$WT$};
\node[state] (D) [above right =of A] {$AT$};
\path (A) edge [bend right =25] node[below =0 cm, left =0.2 cm] {$m=n$} (C);
\path (B) edge [bend left =25] node[below =0 cm, right =0.2 cm] {$m=n$} (C);
\path (D) edge [bend right =25] node[below =0 cm, left =0.2 cm] {$\forall i _{1\leq i \leq m} (f_{i}=g_{i} )$} (A);
\path (D) edge [bend left =25] node[below =0 cm, right =0.2 cm] {$\forall i (f_{i}= 0 \leftrightarrow g_{i}=0)$} (B);
\path (D) edge [bend right =50] node[very near start] {$m=n$} (C);
\path (A) edge [bend left =25] node[below =0.2 cm] {$\forall i (f_{i}= 0 \leftrightarrow g_{i}=0)$} (B);
\path (B) edge [bend left =25] node[above =0.2 cm] {$\forall i _{1\leq i \leq m} (f_{i}=g_{i} )$} (A);
\end{tikzpicture}
}
\caption{Complementary conditions for inverse entailment relations between notions of Trust}
\end{subfigure}
\caption{Entailment relations between notions of Trust and complementary conditions}
\label{fig:entail}
\end{figure}
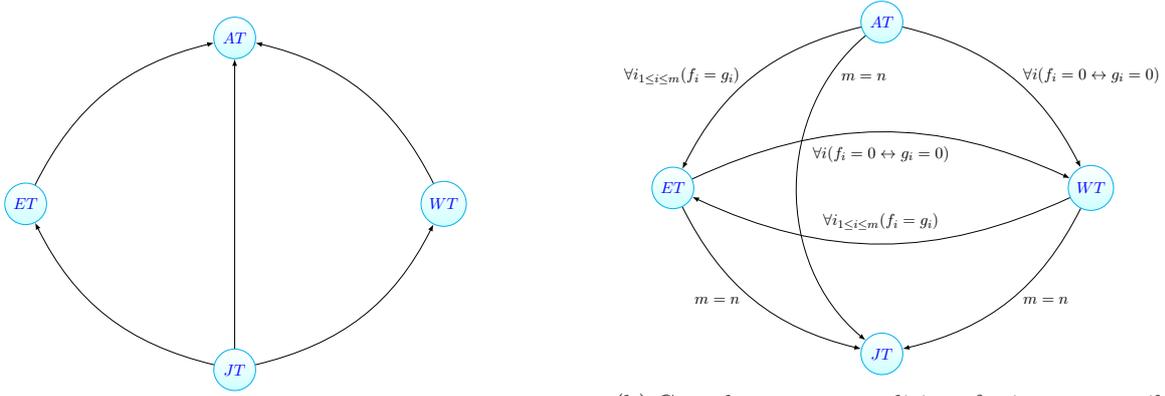

\subsection{Compositions of Trust relations}
\label{subsec:composition}

Once trustworthiness relations are defined in terms of one another, it is possible to investigate what happens when such relations are composed.
As an example, let us consider two systems (called `the originals') such that each of them is a Justifiably Trustworthy copy of the other.
Moreover, let us make a further copy of each of these original systems, such that each of these new copies is only Weakly Trustworthy with respect to its original.
What can we say about the relation of trust between the new copies?
In this subsection, we will discuss these kinds of question.
This is useful in terms of the preservation of trustworthiness under the modular composition of ML systems.
The conclusions that we will reach are mostly negative: by composing Trust relations, we can never be sure of raising the degree of trustworthiness of the copies.

Let us first present a positive (although quite weak) result about the composition of Justifiably, Equally, and Almost Trust relations.

\begin{theorem}
$(\textgoth{a}_{0} \overset{JT}{{\sim}} \textgoth{b}_{0} \land \textgoth{a}_{1} \overset{ET}{{\underset{m}{\sim}}} \textgoth{a}_{0} \land \textgoth{b}_{1} \overset{AT}{{\underset{m}{\leadsto}}} \textgoth{b}_{0}) \rightarrow  \textgoth{b}_{1} \overset{AT}{{\underset{m}{\leadsto}}} \textgoth{a}_{1} $
\end{theorem}

\begin{proof}
By contradiction, let us assume that $\textgoth{b}_{1}$ is not Almost Trustworthy with respect to $\textgoth{a}_{1}$, regarding $ \alpha ^{1}, \ldots , \alpha ^{m} $.
This means that for some $i$ s.t. $ 1 \leq i \leq m $, $\textgoth{b}_{1}$ attributes to value $ \alpha ^{i} $ a lower probability than $\textgoth{a}_{1}$.
Moreover, since $ \textgoth{a}_{0} $ and $ \textgoth{a}_{1} $ are Equally Trustworthy regarding the same values, and $ \textgoth{a}_{0} $ and $ \textgoth{b}_{0} $ are Justifiably Trustworthy, $\textgoth{b}_{0}$ attributes to value $ \alpha ^{i} $ the same probability than $\textgoth{a}_{1}$.
Hence, $\textgoth{b}_{1}$ attributes to value $ \alpha ^{i} $ a lower probability than $\textgoth{b}_{0}$.
However, since $\textgoth{b}_{1}$ is Almost Trustworthy with respect to $\textgoth{b}_{0}$ regarding the same values $ \alpha ^{1}, \ldots , \alpha ^{m} $, $\textgoth{b}_{1}$ attributes to value $ \alpha ^{i} $ the same or an higher probability than $\textgoth{b}_{0}$: contradiction.

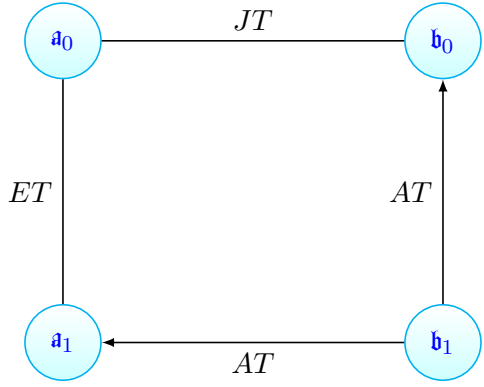
\begin {center}
\begin {tikzpicture}[-latex ,auto ,node distance =4 cm and 5cm ,on grid ,
semithick ,
state/.style ={ circle ,top color =white , bottom color = processblue!20 ,
draw,processblue , text=blue , minimum width =1 cm}]
\node[state] (A) {$\textgoth{a}_{1}$};
\node[state] (B) [right=of A] {$\textgoth{b}_{1} $};
\node[state] (C) [above =of A] {$\textgoth{a}_{0} $};
\node[state] (D) [above =of B] {$\textgoth{b}_{0} $};
\path (B) edge node {$AT$} (A);
\path (B) edge node {$AT$} (D);
\path[-] (C) edge node {$JT$} (D);
\path[-] (A) edge node {$ET$} (C);

\end{tikzpicture}
\end{center}

\end{proof}

Let us now move to the negative results.
We will first show that when we weaken a Trust relation, copying in an Almost Trustworthy way two systems that are Justifiably Trustworthy with respect to one another, we can keep copying an infinite number of times without obtaining Justifiably Trustworthy copies never again.

\begin{theorem}[Infinite Diverging \textbf{AT} Chains]
\label{theorem:divergingATchanins}
Given $\textgoth{a}_{0}$ and $\textgoth{b}_{0}$, such that $\textgoth{a}_{0} \overset{JT}{{\sim}} \textgoth{b}_{0}$, and given $ m \neq n $, such that $\exists k _{m+1\leq k \leq n} (f^{0}_{k}\neq 0)$ (with $f^{0}_{k}$ the frequency of the k-est value for $\textgoth{a}_{0}$), there are two infinite descending chains $\textgoth{a}_{0},\textgoth{a}_{1}, ...$ and $\textgoth{b}_{0},\textgoth{b}_{1}, ...$ such that $\forall i (\textgoth{a}_{i+1} \overset{AT}{{\underset{m}{\leadsto}}} \textgoth{a}_{i} \land \textgoth{b}_{i+1} \overset{AT}{{\underset{m}{\leadsto}}} \textgoth{b}_{i})$ and $ \forall i_{i\neq 0} (\textgoth{a}_{i} \overset{JT}{{\not\sim}} \textgoth{b}_{i}) $.\footnote{
Of course, with $\textgoth{a} \overset{JT}{{\not\sim}} \textgoth{b}$ we mean that $\textgoth{a}$ is not Justifiably Trustworthy with respect to $\textgoth{b}$.
}
\end{theorem}

\begin{proof}
First of all, let us note that if $ f_{k}^{0}\neq 0 $, then $ g_{k}^{0}\neq 0 $ (in fact $ f_{k}^{0} = g_{k}^{0} $), since $\textgoth{a}_{0} \overset{JT}{{\sim}} \textgoth{b}_{0}$.
Then, we decrease the frequency of the values $ \alpha ^{m+1}, ..., \alpha ^{n} $ to increase that of $ \alpha ^{1}, ..., \alpha ^{m} $.

Let us take the smallest $k$ such that $m+1\leq k \leq n$ and $f_{k}^{0} \neq 0$.
Given $\textgoth{a}_{0} \overset{JT}{{\sim}} \textgoth{b}_{0}$, $f_{k}^{0} = g_{k}^{0} \neq 0$

Let us now define:
\medskip
\[
f^{i+1}_{1}= f^{i}_{1} + \frac{1}{2} f^{i}_{k}
\]
\[
g^{i+1}_{1}= g^{i}_{1} + \frac{1}{3} g^{i}_{k}
\]

\medskip
\noindent
where $f^{i}_{k}$ and $g^{i}_{k}$ are such that $\textgoth{a}_{i}: \alpha ^{k}_{f^{i}_{k}}$ and $\textgoth{b}_{i}: \alpha ^{k}_{g^{i}_{k}}$.

Of course, frequencies of the other outputs ($ \alpha ^{m+1}, ..., \alpha ^{n} $) will decrease accordingly.
More precisely:
\medskip
\[
f^{i+1}_{m+k}= \frac{1}{2} f^{i}_{k}
\]
\[
g^{i+1}_{m+k}= \frac{2}{3} g^{i}_{k}
\]

\medskip
\noindent
Nonetheless, this is never a problem for \textbf{AT}.

A lower point in the descending chain so defined is always Almost Trustworthy with respect to its upper point.
Indeed, as for $\alpha ^{1}$, $f^{i+1}_{1}\geq f^{i}_{1}$ and $g^{i+1}_{1}\geq g^{i}_{1}$.
As for $\alpha ^{2}, ...,  \alpha ^{m} $, $f^{i+1}_{1} = f^{i}_{1}$ and $g^{i+1}_{1} = g^{i}_{1}$.

Moreover, the chains never converge for finite $i$.
Indeed, for every $i$ such that $ i \neq 0 $, $f^{i}_{1} \neq g^{i}_{1}$.

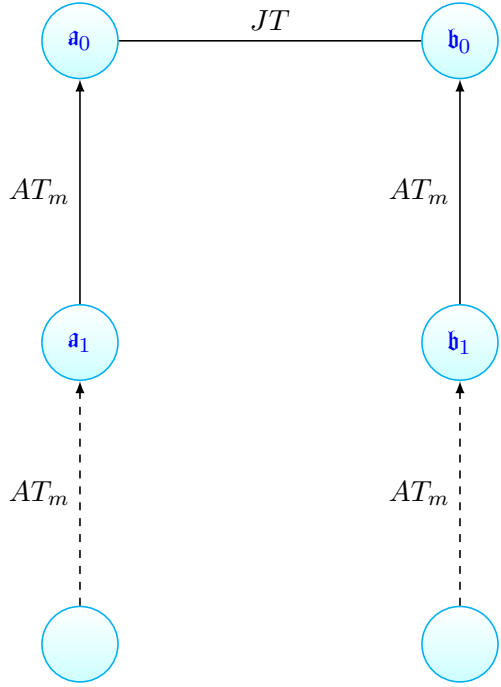
\begin {center}
\begin {tikzpicture}[-latex ,auto ,node distance =4 cm and 5cm ,on grid ,
semithick ,
state/.style ={ circle ,top color =white , bottom color = processblue!20 ,
draw,processblue , text=blue , minimum width =1 cm}]
\node[state] (A) { };
\node[state] (B) [right=of A] { };
\node[state] (C) [above =of A] {$\textgoth{a}_{1}$};
\node[state] (D) [above =of B] {$\textgoth{b}_{1}$};
\node[state] (E) [above =of C] {$\textgoth{a}_{0}$};
\node[state] (F) [above =of D] {$\textgoth{b}_{0}$};

\draw[every to/.style={append after command={[draw,dashed]}}] (A) to node {$AT_{m}$} (C);
\path (C) edge node {$AT_{m}$} (E);
\draw[every to/.style={append after command={[draw,dashed]}}] (B) to node {$AT_{m}$} (D);
\path (D) edge node {$AT_{m}$} (F);
\path[-] (E) edge node {$JT$} (F);

\end{tikzpicture}
\end{center}

\end{proof}

Let us now see some generalizations of theorem~\ref{theorem:divergingATchanins}, obtained by weakening the \textit{unpreserved} Trust relation (corollary~\ref{corollary:NonETATchanins}), strengthening the Trust relations composing the infinite descending chains (corollaries~\ref{corollary:WT} and \ref{corollary:ET}), or both (corollary~\ref{corollary:NonETWTchanins}).

\begin{corollary}[Infinite non-\textbf{ET} \textbf{AT} Chains]
\label{corollary:NonETATchanins}
Given $\textgoth{a}_{0}$ and $\textgoth{b}_{0}$, such that $\textgoth{a}_{0} \overset{JT}{{\sim}} \textgoth{b}_{0}$, and given $ m \neq n $, such that $\exists k _{m+1\leq k \leq n} (f^{0}_{k}\neq 0)$ (with $f^{0}_{k}$ the frequency of the k-est value for $\textgoth{a}_{0}$), there are two infinite descending chains $\textgoth{a}_{0},\textgoth{a}_{1}, ...$ and $\textgoth{b}_{0},\textgoth{b}_{1}, ...$ such that $\forall i (\textgoth{a}_{i+1} \overset{AT}{{\underset{m}{\leadsto}}} \textgoth{a}_{i} \land \textgoth{b}_{i+1} \overset{AT}{{\underset{m}{\leadsto}}} \textgoth{b}_{i})$ and $ \forall i_{i\neq 0} (\textgoth{a}_{i} \overset{ET}{{\underset{m}{\not\sim}}} \textgoth{b}_{i}) $.
\end{corollary}

\begin{proof}
The proof is the same of theorem~\ref{theorem:divergingATchanins}.
Just note that $\forall i _{i>0}  (f^{i}_{1}\neq g^{i}_{1})$ and $ 1 \leq m$.
\end{proof}

\begin{corollary}[Infinite Diverging \textbf{WT} Chains]
\label{corollary:WT}
Given $\textgoth{a}_{0}$ and $\textgoth{b}_{0}$, such that $\textgoth{a}_{0} \overset{JT}{{\sim}} \textgoth{b}_{0}$ and $\exists l (f^{0}_{l}\neq 0)$, and given $ m \neq n $, such that $\exists k _{k \neq l \land m+1\leq k \leq n} (f^{0}_{k}\neq 0)$, there are two infinite descending chains $\textgoth{a}_{0},\textgoth{a}_{1}, ...$ and $\textgoth{b}_{0},\textgoth{b}_{1}, ...$ such that $\forall i (\textgoth{a}_{i+1} \overset{WT}{{\underset{m}{\leadsto}}} \textgoth{a}_{i} \land \textgoth{b}_{i+1} \overset{WT}{{\underset{m}{\leadsto}}} \textgoth{b}_{i})$ and $ \forall i_{i\neq 0} (\textgoth{a}_{i} \overset{JT}{{\not\sim}} \textgoth{b}_{i}) $.
\end{corollary}

\begin{proof}
The proof is the same of theorem~\ref{theorem:divergingATchanins} but with $ \alpha _{l} $ in place of $ \alpha _{1} $.
Just observe that at each step the only frequencies that decrease are $ f^{i}_{k} $ and $g^{i}_{k}$, which nonetheless never reach zero for $ i \in \mathbb{N}$, and are such that $ m<k $.
Moreover, the only frequencies that increase are $ f^{i}_{l} $ and $g^{i}_{l}$, and are such that for every $ i \in \mathbb{N} $, $ f^{i}_{l} \neq 0 $ and $g^{i}_{l} \neq 0$.\footnote{
Note that $ l $ can be less than or greater than $ m $, and $ \textgoth{a}_{i+1} $ and $ \textgoth{b}_{i+1} $ are still \textbf{WT} regarding $ \textgoth{a}_{i} $ and $ \textgoth{b}_{i} $ respectively.
}
\end{proof}

\begin{corollary}[Infinite non-\textbf{ET} \textbf{WT} Chains]
\label{corollary:NonETWTchanins}
Given $\textgoth{a}_{0}$ and $\textgoth{b}_{0}$, such that $\textgoth{a}_{0} \overset{JT}{{\sim}} \textgoth{b}_{0}$ and $\exists l _{1\leq l \leq m} (f^{0}_{l}\neq 0)$, and given $ m \neq n $, such that $\exists k _{m+1\leq k \leq n} (f^{0}_{k}\neq 0)$, there are two infinite descending chains $\textgoth{a}_{0},\textgoth{a}_{1}, ...$ and $\textgoth{b}_{0},\textgoth{b}_{1}, ...$ such that $\forall i (\textgoth{a}_{i+1} \overset{WT}{{\underset{m}{\leadsto}}} \textgoth{a}_{i} \land \textgoth{b}_{i+1} \overset{WT}{{\underset{m}{\leadsto}}} \textgoth{b}_{i})$ and $ \forall i_{i\neq 0} (\textgoth{a}_{i} \overset{ET}{{\underset{m}{\not\sim}}} \textgoth{b}_{i}) $.
\end{corollary}

\begin{proof}
The proof is the same as for the corollary~\ref{corollary:WT}.
Just observe that in this case $ l $ cannot be greater than $ m $ (hence, the change in the assumptions of the statement).
\end{proof}

\begin{corollary}[Infinite Diverging \textbf{ET} Chains]
\label{corollary:ET}
Given $\textgoth{a}_{0}$ and $\textgoth{b}_{0}$, such that $\textgoth{a}_{0} \overset{JT}{{\sim}} \textgoth{b}_{0}$, and given $ m \neq n $, such that $\exists k,l _{k \neq l \land m+1\leq l \leq n \land m+1\leq k \leq n} (f^{0}_{k}\neq 0 \land f^{0}_{l}\neq 0)$, there are two infinite descending chains $\textgoth{a}_{0},\textgoth{a}_{1}, ...$ and $\textgoth{b}_{0},\textgoth{b}_{1}, ...$ such that $\forall i (\textgoth{a}_{i+1} \overset{ET}{{\underset{m}{\sim}}} \textgoth{a}_{i} \land \textgoth{b}_{i+1} \overset{ET}{{\underset{m}{\sim}}} \textgoth{b}_{i})$ and $ \forall i_{i\neq 0} (\textgoth{a}_{i} \overset{JT}{{\not\sim}} \textgoth{b}_{i}) $.
\end{corollary}

\begin{proof}
We just replicate the proof procedure of theorem~\ref{theorem:divergingATchanins} but taking the frequencies $ f_{k} $, $ g_{k} $, $ f_{l} $ and $ g_{l} $.
\end{proof}

\section{Trust and Logical Rules}
\label{sec:TrustLogic}

In section~\ref{sec:TrustCopies} and specifically subsection~\ref{subsec:DiffTrust}, we have defined local notions of trustworthiness (Justifiably, Weak, Equal and Almost Trust) for applications of ML systems to single rows $ \sigma $ of the Test Set, and regarding a specific target atomic variable $ a $.
In this section, we investigate whether and how trustworthiness is preserved when logical rules are applied to judgments regarding ML systems.
In other words, we wish to check under which logical composition of atomic judgments regarding ML systems trustworthiness is preserved.
In order to do so, we define in subsection~\ref{subsec:TrustMultAppl} the generalizations of the notions of trustworthiness for sets of lists of values assignments $\Sigma$ and sets of atomic target variables $\mathscr{A}$.
Then, in subsection~\ref{subsec:LogicTrust} we use this generalization to investigate logically complex judgments.

\subsection{General Trust for $ \Sigma $ and $ \mathscr{A} $}
\label{subsec:TrustMultAppl}

In this section, we define trustworthiness for a set of lists of values assignments $\Sigma$ and a set of atomic target variables $\mathscr{A}$, generalizing the local notions that we have developed in section~\ref{sec:TrustCopies}.
The generalization of Justifiably Trust is straightforward, since it regards preservation of all the possible values of the target variable:

\begin{definition}[General Justifiably Trustworthy Copies]
\label{JTGeneral}
Given two ML systems dubbed respectively the \textit{original} and the \textit{copy}, we say that the \textit{copy} is Justifiably Trustworthy with respect to the \textit{original} regarding a set of lists of values assignments $\Sigma$ and a set of atomic target variables $\mathscr{A}$ iff for every $\sigma \in \Sigma $ and $a \in \mathscr{A} $ the \textit{copy} is Justifiably Trustworthy with respect to the \textit{original} considering the feature $ a $ and the list of values assignments $ \sigma $.
\end{definition}

On the contrary, Weakly, Equally and Almost Trust are more problematic.
Indeed, when we defined the local versions of these notions for single values attributions $ \sigma $ and single atomic target variables $ a $, we specified that a copy is Weakly, Equally and Almost Trustworthy \textit{regarding a set of possible outputs} $ \alpha ^{1}, \ldots , \alpha ^{m}$ for $ a $.
Now, if we take into account a set $\mathscr{A}$ of atomic target variables, we need a way to specify the relevant values for each of its elements.
To do so, we will use a function $v$ that ascribes to each $a\in \mathscr{A}$ a set of possible values $ \alpha ^{1}, \ldots ,  \alpha ^{m}$ of $ a $.
Hence, the definition becomes:

\begin{definition}[General Weakly, Equally and Almost Trustworthy Copies]
\label{WEATGeneral}
Given two ML systems dubbed respectively the \textit{original} and the \textit{copy}, we say that the \textit{copy} is Weakly, Equally or Almost Trustworthy with respect to the \textit{original} regarding a set of lists of values assignments $\Sigma$, a set of atomic target variables $\mathscr{A}$ and a function $v$ assigning relevant values to each atomic target variable $a\in \mathscr{A}$ iff for every $\sigma \in \Sigma $ and $a \in A $ the \textit{copy} is respectively Weakly, Equally or Almost Trustworthy with respect to the \textit{original} considering the feature $ a $ and the list of values assignments $ \sigma $, regarding values $ v(a) $.
\end{definition}

\begin{example}[Trust over $\Sigma$, $\mathscr{A}$ and $v$]
\label{exa:AtTrustSigma}
Let us consider a pair of diagnostic systems, labeled the \textit{original} and the \textit{copy}, devised both for the detection of some set of diseases such as \textit{Chickenpox} and \textit{Hepatitis}, and for the evaluation of risk factors for skin cancer.
The systems consider a different variable for each disease and give the probabilities for its severity, with the following possible outputs:  $absent$, $minor$, $moderate$, $major$, and $extreme$.
Moreover, the system gives an evaluation of the risk of developing skin cancer (corresponding to the atomic variable $SkinC$) with the following possible outputs: $no$ (denoting the probability that the patient will not develop a tumor in the foreseeable future), $more~2~years$ (denoting the probability that the patient will develop a tumor in more than two years), and $less~2~years$ (denoting the probability that the patient will develop a tumor in less than two years).
Let us, moreover, assume that $\Sigma$ contains all the data regarding $n$ patients described respectively in the lists of values assignments $\sigma_{1} , \ldots , \sigma_{n}$.
In conclusion, let as assume that function $v$ assigns $moderate$, $major$, and $extreme$ to both variables $Chickenpox$ and $Hepatitis$ and $less~2~years$ to $SkinC$.
This means that we want to focus on the probabilities attached to severe cases of \textit{Chickenpox} and \textit{Hepatitis} and to the development of skin cancer in two years, prioritizing worst scenarios.

The \textit{copy} is Justifiably Trustworthy with respect to the \textit{original} regarding the patients described in $\Sigma$ and the set of atomic target variables $\mathscr{A} $ $=$ $\{ Chickenpox, $ $Hepatitis,$ $ SkinC \}$ iff for every $\sigma _{i} \in \Sigma $ (that is for every patient) and $a \in \mathscr{A} $ the \textit{copy} is Justifiably Trustworthy with respect to the \textit{original} considering the feature $ a $ and the list of values assignments $ \sigma _{i}$.
That is, Justifiably Trustworthiness ignores function $v$, and asks that original and copy behaves identically for all possible outputs of $Chickenpox$, $Hepatitis$, and $SkinC$, for every patient described in $\Sigma$. 

The \textit{copy} is Weakly, Equally or Almost Trustworthy with respect to the \textit{original} regarding the patients described in $\Sigma$, the set of atomic target variables $\mathscr{A} $ $=$ $\{ Chickenpox, $ $Hepatitis,$ $ SkinC \}$ and the function $v$ assigning relevant values to each atomic target variable $a\in \mathscr{A}$ --- $v(Chickenpox)=v(Hepatitis)= \{moderate, major, extreme \}$ and $v(SkinC)= \{ less~2~years \}$ --- iff for every $\sigma \in \Sigma $ and $a \in A $ the \textit{copy} is respectively Weakly, Equally or Almost Trustworthy with respect to the \textit{original} considering the feature $a$ and the list of values assignments $ \sigma $, regarding values $v(a)$.

As an example, the \textit{copy} is Almost Trustworthy iff, for every patient, it attributes at least the same probability as the original to the judgments that ascribes them $moderate$, $major$, or $extreme$ case of $Chickenpox$ and $Hepatitis$, and the risk of contracting a skin cancer within the next two years.
In other words, let us consider the following judgments for every patient $\sigma _{i}$:\footnote{
We use $f^{o}$, $g^{o}$, and $h^{o}$ to indicate the probability given by the original system and $f^{c}$, $g^{c}$, and $h^{c}$ for the probability given by the copy.
}
\medskip
\[
\sigma _{i} \rhd Chickenpox : absent _{f^{o}_{i,1}}, minor_{f^{o}_{i,2}}, moderate_{f^{o}_{i,3}}, major_{f^{o}_{i,4}}, extreme_{f^{o}_{i,5}}
\]
\[
\sigma _{i} \rhd Chickenpox : absent _{f^{c}_{i,1}}, minor_{f^{c}_{i,2}}, moderate_{f^{c}_{i,3}}, major_{f^{c}_{i,4}}, extreme_{f^{c}_{i,5}}
\]
\[
\sigma _{i} \rhd Hepatitis : absent _{g^{o}_{i,1}}, minor_{g^{o}_{i,2}}, moderate_{g^{o}_{i,3}}, major_{g^{o}_{i,4}}, extreme_{g^{o}_{i,5}}
\]
\[
\sigma _{i} \rhd Hepatitis : absent _{g^{c}_{i,1}}, minor_{g^{c}_{i,2}}, moderate_{g^{c}_{i,3}}, major_{g^{c}_{i,4}}, extreme_{g^{c}_{i,5}}
\]
\[
\sigma _{i} \rhd SkinC : no _{h^{o}_{i,1}}, more~2~years _{h^{o}_{i,2}}, less~2~years{h^{o}_{i,3}}
\]
\[
\sigma _{i} \rhd SkinC : no _{h^{c}_{i,1}}, more~2~years _{h^{c}_{i,2}}, less~2~years{h^{c}_{i,3}}
\]

\medskip

\noindent
The requirements for Almost Trustworthiness of the copy are that $ h^{o}_{i,3} \leq h^{c}_{i,3}  $, and for $3 \leq j \leq 5$, $ f^{o}_{i,j} \leq f^{c}_{i,j}$ and $ g^{o}_{i,j} \leq g^{c}_{i,j}$.
Note that, assuming that the original system is reliable, we could be satisfied with our copy slightly overestimating these probabilities, as remarked in example~\ref{exa:WT}.
\end{example}

\subsection{Logical rules and trustworthiness}
\label{subsec:LogicTrust}

\subsubsection{Definitions for Copies of Non-Atomic Systems}

In section~\ref{sec:TrustCopies}, we have defined Trust notions for copies of applied ML systems taking into account only atomic outputs for atomic target variables.
In this section, we will investigate how these notions generalize when non-atomic variables and outputs are considered.
In order to do so, it is convenient to define the notion of logical construction and deconstruction of an ML system:

\begin{definition}[Logical Construction of ML systems]
Given a set $\textgoth{A}$ of ML systems $\textgoth{a}_{1},\ldots ,\textgoth{a}_{n}$, we say that an ML system $\textgoth{b}$ is a logical construction of $\textgoth{a}_{1}, \ldots , \textgoth{a}_{n}$ ($\textgoth{a}_{1}, \ldots , \textgoth{a}_{n} \mapsto ^{cstr} \textgoth{b}$) iff there are values $\beta ^{1}, \ldots \beta ^{n}$ and $\delta$ respectively for $\textgoth{a}_{1}, \ldots , \textgoth{a}_{n}$ and $\textgoth{b}$ s.t. for some probabilities $f_{1}, \ldots , f_{n},h$, the judgment $\textgoth{b}: \delta _{h}$ can be derived from $\textgoth{a}_{1}:\beta^{1} _{f_{1}}, \ldots , \textgoth{a}_{n}: \beta^{n}_{f_{n}}$ using only right I-rules.
Sometimes we will speak directly of $\textgoth{b}: \delta _{h}$ as a logical construction of $\textgoth{a}_{1}:\beta^{1} _{f_{1}} , \ldots , \textgoth{a}_{n}: \beta^{n}_{f_{n}}$ ($\textgoth{a}_{1}:\beta^{1} _{f_{1}} , \ldots , \textgoth{a}_{n}: \beta^{n}_{f_{n}} \mapsto ^{cstr} \textgoth{b}: \delta _{h}$).
\end{definition}

\begin{definition}[Logical Deconstruction of ML systems]
Given a set $\textgoth{A}$ of ML systems $\textgoth{a}_{1},\ldots ,\textgoth{a}_{n}$, we say that an ML system $\textgoth{b}$ is a logical construction of $\textgoth{a}_{1}, \ldots , \textgoth{a}_{n}$ ($\textgoth{a}_{1}, \ldots , \textgoth{a}_{n} \mapsto ^{decstr} \textgoth{b}$) iff there are values $\beta ^{1}, \ldots \beta ^{n}$ and $\delta$ respectively for $\textgoth{a}_{1}, \ldots , \textgoth{a}_{n}$ and $\textgoth{b}$ s.t. for some probabilities $f_{1}, \ldots , f_{n},h$, the judgment $\textgoth{b}: \delta _{h}$ can be derived from $\textgoth{a}_{1}:\beta^{1} _{f_{1}}, \ldots , \textgoth{a}_{n}: \beta^{n}_{f_{n}}$ using only right E-rules.
Sometimes we will speak directly of $\textgoth{b}: \delta _{h}$ as a logical construction of $\textgoth{a}_{1}:\beta^{1} _{f_{1}} , \ldots , \textgoth{a}_{n}: \beta^{n}_{f_{n}}$ ($\textgoth{a}_{1}:\beta^{1} _{f_{1}} , \ldots , \textgoth{a}_{n}: \beta^{n}_{f_{n}} \mapsto ^{decstr} \textgoth{b}: \delta _{h}$).
\end{definition}

Note that the starting systems $\textgoth{a}_{1},\ldots ,\textgoth{a}_{n}$ in logical construction and the system $\textgoth{b}$ reached in logical deconstruction can have non-atomic target variables and deal with non-atomic values.
However, logical constructions starting with systems dealing only with atomic target variables and atomic values, and logical deconstructions reaching systems dealing only with atomic target variables and atomic values are special cases of logical construction and deconstruction.
Hence, if trust is preserved across logical construction and deconstruction, trust for atomic variables and values entails trust for non-atomic variables and values, and vice-versa.

Note also that logical constructions and logical deconstructions of a set of ML systems do not count as copies.
Indeed, we defined a copy of an ML system as a different ML system (different Learning Algorithm or different Training Set) but possibly dealing with \textit{the same Test Set} and having the \textit{same target variable}.
On the contrary, as clarified in subsection~\ref{subsub:ExtLang}, when we apply logical rules (and so also in a logical construction or deconstruction), we refer to the same system (same Learning Algorithm and Training Set) but change the target variable of the query.
Moreover, the result of construction or deconstruction will not necessarily use the same row of the Test Set; as an example, consider the rules of introduction for $ \times $.

Before addressing the preservation of trustworthiness under logical construction and deconstruction, we need to define the notions of Trust for non-atomic target variables and values.
For Justifiably Trust, the definition is as follows: 

\begin{definition}[Justifiably Trustworthy For Non-Atomic Elements]
Given two ML systems $ \textgoth{a} $ and $\textgoth{b}$ both having $t$ as target variable, dubbed respectively the \textit{original} and the \textit{copy}, differing for either the Learning Algorithms they implement or the Data Set used during training, we say that the \textit{copy} is Justifiably Trustworthy with respect to the \textit{original} considering the variable $ t $ and the same row of values assignments $ \sigma $ from the same Test Set as the original iff for every possible value $ \beta $ for $ t $, $ \textgoth{a}:\beta _{f} $ iff $ \textgoth{b}:\beta _{f} $.
\end{definition}

The only difficulty with this definition is that we cannot list all the possible non-atomic values of a non-atomic variable, as opposed to what we did with atomic values of atomic variables.
Hence, since the number of possible outputs for $ t $ is not finite, we cannot check JT by examining all of them one by one.
However, this is not a problem, since we will focus on proving the preservation of trustworthiness when logical rules are applied, and not on proving JT directly for logically complex judgments.
On the contrary, focusing on such preservation is the easiest approach to deal with JT for non-atomic variables and values.

As for the generalization of the previous definition for a set of lists of values assignments $\Sigma$ and a set of  variables $\mathscr{T}$, we can easily adapt definition~\ref{JTGeneral}:

\begin{definition}[General Non-Atomic Justifiably Trustworthy Copies]
\label{LogicJTGeneral}
Given two ML systems dubbed respectively the \textit{original} and the \textit{copy}, we say that the \textit{copy} is Justifiably Trustworthy with respect to the \textit{original} regarding a set of lists of values assignments $\Sigma$ and a set of target variables $\mathscr{T}$ iff for every $\sigma \in \Sigma $ and $t \in \mathscr{T} $ the \textit{copy} is Justifiably Trustworthy with respect to the \textit{original} considering the variable $ t $ and the list of values assignments $ \sigma $.
\end{definition}

The problem with Equally and Almost Trust is that we need to specify a subset of outputs for the target variable $t$ that we want to consider relevant, just like in subsection~\ref{subsec:DiffTrust} we picked outputs $ \alpha ^{1}, \ldots , \alpha ^{m} $.
In this case, we have to focus on a (possibly infinite) set $ B $ of possible values for $t$:

\begin{definition}[Equally Trustworthy Copies for Non-Atomic Variables]
Given two ML systems $ \textgoth{a} $ and $\textgoth{b}$ both having $t$ as target variable, dubbed respectively the \textit{original} and the \textit{copy},  differing for either the Learning Algorithms they implement or the Data Set used during training, we say that the \textit{copy} is Equally Trustworthy with respect to the \textit{original} regarding the set of values $ B $ when considering variable $ t $ and the same row of values assignments $ \sigma $ from the same Test Set as the original iff for every value $ \beta \in B $ for $ t $, $ \textgoth{a}:\beta _{f} $ iff $ \textgoth{b}:\beta _{f} $.
\end{definition}

\begin{definition}[Almost Trustworthy Copies for Non-Atomic Variables]
\label{def:ATNoAtm}
Given two ML systems $ \textgoth{a} $ and $\textgoth{b}$ both having $t$ as target variable, dubbed respectively the \textit{original} and the \textit{copy}, differing for either the Learning Algorithms they implement or the Data Set used during training, we say that the \textit{copy} is Almost Trustworthy with respect to the \textit{original} regarding the set of values $ B $ when considering variable $ t $ and the same row of values assignments $ \sigma $ from the same Test Set as the original iff for every value $ \beta \in B $ for $ t $, if $ \textgoth{a}:\beta _{f} $ and $ \textgoth{b}:\beta _{g} $, then $ g\geq f $.
\end{definition}

Also in this case, the generalization of the previous definition for a set of lists of values assignments $\Sigma$ and a set of variables $\mathscr{T}$ can be adapted from definition~\ref{WEATGeneral}:

\begin{definition}[General Non-Atomic Equally (and Almost) Trustworthy Copies]
\label{EATGeneral}
Given two ML systems dubbed respectively the \textit{original} and the \textit{copy}, we say that the \textit{copy} is Equally or Almost Trustworthy with respect to the \textit{original} regarding a set of lists of values assignments $\Sigma$, a set of target variables $\mathscr{T}$ and a function $v$ assigning relevant values to each atomic target variable $t\in \mathscr{T}$ iff for every $\sigma \in \Sigma $ and $t \in T $ the \textit{copy} is respectively Equally or Almost Trustworthy with respect to the \textit{original} considering the variable $ t $ and the list of values assignments $ \sigma $, regarding the set of values $ B = v(t) $.
\end{definition}

As for Weak Trust, we can adapt definition~\ref{def:ATNoAtm}, by adding the condition of the zeros:

\begin{definition}[Weakly Trustworthy Copies for Non-Atomic Variables]
\label{def:WTNoAtm}
Given two ML systems $ \textgoth{a} $ and $\textgoth{b}$ both having $t$ as target variable, dubbed respectively the \textit{original} and the \textit{copy}, differing for either the Learning Algorithms they implement or the Data Set used during training, we say that the \textit{copy} is Weakly Trustworthy with respect to the \textit{original} regarding the set of values $ B $ when considering variable $ t $ and the same row of values assignments $ \sigma $ from the same Test Set as the original iff:
\begin{itemize}
\item for every value $ \delta $ for $t$, $ \textgoth{a}:\delta _{0} $ iff $ \textgoth{b}:\delta _{0} $;
\item for every value $ \beta \in B $ for $ t $, if $ \textgoth{a}:\beta _{f} $ and $ \textgoth{b}:\beta _{g} $, then $ g\geq f $.
\end{itemize} 
\end{definition}

\noindent
Note that the first condition regards every possible value for $ t $, while the second condition applies only to values in $B$.

Also in this case, the generalization of the previous definition for a set of lists of values assignments $\Sigma$ and a set of variables $\mathscr{T}$ can be adapted from definition~\ref{WEATGeneral}:

\begin{definition}[General Non-Atomic Weakly Trustworthy Copies]
\label{WTGeneral}
Given two ML systems dubbed respectively the \textit{original} and the \textit{copy}, we say that the \textit{copy} is Weakly Trustworthy with respect to the \textit{original} regarding a set of lists of values assignments $\Sigma$, a set of target variables $\mathscr{T}$ and a function $v$ assigning relevant values to each atomic target variable $t\in \mathscr{T}$ iff for every $\sigma \in \Sigma $ and $t \in T $ the \textit{copy} is Weakly Trustworthy with respect to the \textit{original} considering the variable $ t $ and the list of values assignments $ \sigma $, regarding the set of values $ B = v(t) $.
\end{definition}

\subsubsection{Preservation Results}
\label{subsubsec:PreservationResults}

Let us now address trustworthiness for possibly non-atomic variables and values, investigating how it is linked to trustworthiness for atomic variables and values, starting with the generalization of JT.
Let us first prove a simple lemma:

\begin{lemma}[Construction and Deconstruction Preserve Identity]
\label{lemma:CompPresAgree}
If $\textgoth{a}_{1}:\beta^{1} _{f_{1}} , \ldots , \textgoth{a}_{n}: \beta^{n}_{f_{n}} \mapsto ^{cstr} \textgoth{c}: \delta _{g}$ and $\textgoth{b}_{1}:\beta^{1} _{f_{1}} , \ldots , \textgoth{b}_{n}: \beta^{n}_{f_{n}} \mapsto ^{cstr} \textgoth{d}: \delta _{h}$, or $\textgoth{a}_{1}:\beta^{1} _{f_{1}} , \ldots , \textgoth{a}_{n}: \beta^{n}_{f_{n}} \mapsto ^{decstr} \textgoth{c}: \delta _{g}$ and $\textgoth{b}_{1}:\beta^{1} _{f_{1}} , \ldots , \textgoth{b}_{n}: \beta^{n}_{f_{n}} \mapsto ^{decstr} \textgoth{d}: \delta _{h}$, and moreover the target variable of $\textgoth{d}$ is identical to the target variable of $\textgoth{c}$, then $g=h$.
\end{lemma}

\begin{proof}
Let us prove that construction preserves identity.
    If $\textgoth{a}_{1}:\beta^{1} _{f_{1}} , \ldots , \textgoth{a}_{n}: \beta^{n}_{f_{n}} \mapsto ^{cstr} \textgoth{c}: \delta _{g}$ and $\textgoth{b}_{1}:\beta^{1} _{f_{1}} , \ldots , \textgoth{b}_{n}: \beta^{n}_{f_{n}} \mapsto ^{cstr} \textgoth{d}: \delta _{h}$ and the target variable of $\textgoth{d}$ is identical to the target variable of $\textgoth{c}$, then the same rules are applied in the same order to both $\textgoth{a}_{1}:\beta^{1} _{f_{1}} , \ldots , \textgoth{a}_{n}: \beta^{n}_{f_{n}}$ and $\textgoth{b}_{1}:\beta^{1} _{f_{1}} , \ldots , \textgoth{b}_{n}: \beta^{n}_{f_{n}}$, in order to reach $\textgoth{c}: \delta _{g}$ and $\textgoth{d}: \delta _{h}$.
    Indeed, applications of the rules for $\rightarrow$ and $\times$ leave evidence in both the variable and the value, while those of the rules for $\bot$ and $+$ leave evidence in the value, and so the reconstruction of the proof from the variable and the value of the conclusion is univocal.
    Hence, the proof easily follows from the observation that the probability in the conclusion of an I-rule is a function of the probabilities in the premises.
    The proof that deconstruction preserves identity is similar.
\end{proof}

\begin{theorem}[Logical Construction and Deconstruction Preserve JT]
\label{theo:LRulesPreserveJT}
Let us consider two sets $\textgoth{A}$ and $\textgoth{B}$ of applied ML systems and assume that for every $\textgoth{a}_{i} \in \textgoth{A}$ there is $\textgoth{b}_{i} \in \textgoth{B}$ such that $\textgoth{a}_{i}$ is JT with respect to $\textgoth{b}_{i}$.
If $\textgoth{a}_{1}, \ldots \textgoth{a}_{n} \mapsto ^{cstr} \textgoth{c}$, and $\textgoth{b}_{1}, \ldots \textgoth{b}_{n} \mapsto ^{cstr} \textgoth{d}$, or $\textgoth{a}_{1}, \ldots \textgoth{a}_{n} \mapsto ^{decstr} \textgoth{c}$, and $\textgoth{b}_{1}, \ldots \textgoth{b}_{n} \mapsto ^{decstr} \textgoth{d}$, and moreover the target variable of $\textgoth{d}$ is identical to the target variable of $\textgoth{c}$, then for every possible value $\delta$, $\textgoth{c}: \delta _{f}$ iff $\textgoth{d}: \delta _{f}$.
Hence, $ \textgoth{c} $ is JT with respect to $\textgoth{d}$.
\end{theorem}

\begin{proof}
    If $\textgoth{a}_{i}$ is JT with respect to $\textgoth{b}_{i}$, then they share both the same target variable and the same list of values attributions, and they agree on the probability of every possible value for the target variable.
    Hence, the theorem follows from lemma~\ref{lemma:CompPresAgree}.
\end{proof}

Theorem~\ref{theo:LRulesPreserveJT} proves that ML systems that are Justifiably Trustworthy regarding a set of target variables and regarding a set of outputs remain so regarding variables and values that are logically constructed or deconstructed from these.
From this, it follows that we can generalize the trustworthiness of ML systems from only atomic applications to non-atomic ones.
Indeed, as a special case, theorem~\ref{theo:LRulesPreserveJT} remains valid when the systems are originally applied to atomic variables and only for evaluating atomic values.

In general, logical construction changes not only the structure of the target variable but also the list of value attributions in the antecedent of the judgments.
The problem regards the rules for conditional and conjunction in particular.
As for the conditional, there seems to be no solution, since we need to change the antecedent in order to introduce the implication.
The situation is different for the conjunction.
In fact, when the target variables in the premises are mutually independent, the rule of introduction of the conjunction can be simplified as follows, where the extra premise in $ u \ci t $ formalizes the mutual independence of $ u $ and $ t$:\footnote{
Mutual independence is defined as usual: $ A $ and $B$ are mutually independent iff $P(B | A)= P(B)$ (or equivalently $P(A | B)= P(A)$).
Note that this rule is alike the one adopted in \citep{10.1093/logcom/exaf003}.
}
\begin{prooftree}
    \AxiomC{$ \sigma  \rhd u : \delta _{g} $}
	\AxiomC{$ \sigma \rhd t: \beta _{f} $}
    \AxiomC{$ u \ci t $}
	\RightLabel{{\tiny I$\times ^{\ci}$}}
	\TrinaryInfC{$ \sigma  \rhd \langle t, u \rangle: (\beta\times  \delta) _{f\cdot g} $}
\end{prooftree}

\noindent

Hence, we can obtain the following result regarding the preservation of General JT:

\begin{corollary}[Logical Construction Preserves General JT]
\label{PreservationLogicJTGeneral}
Given two ML systems dubbed respectively the \textit{original} and the \textit{copy}, if the \textit{copy} is Justifiably Trustworthy with respect to the \textit{original} regarding a set of lists of values assignments $\Sigma$ and a set of mutually independent atomic target variables $\mathscr{A}$ that do not occur in $\Sigma$, then the \textit{copy} is Justifiably Trustworthy with respect to the \textit{original} also regarding the set of target variables $\mathscr{T}$, obtained closing $ \mathscr{A} $ under conjunction, and the same set of lists of values assignments $\Sigma$.
\end{corollary}

\begin{proof}
Assuming that none of the target variables $\mathscr{A}$ occur in $\Sigma$, the only rule that can introduce a conjunction in the target variable is I$\times ^{\ci}$, which has the same antecedent in both the premises and the conclusion.
Note that this rule is applicable because of the assumption that all variables in $\mathscr{A}$ are mutually independent.
Moreover, the rules for conditional cannot be applied, since conditional does not occur in $\mathscr{T}$.
As for the other rules, it is obvious by inspecting the rules of tables~\ref{tab:triangle} and \ref{tab:LRtriangle} that, when the target variable of every premise occurs in the antecedent of none of the other premises, the antecedents are all equal in both the premises and the conclusion.
As an example, the rules for the introduction of disjunction on the left cannot be applied, since it requires that $t:\gamma$ and $t:\beta$ occur in both the antecedent and the conclusion of the premises.
This fact entails that the ML system remains JT regarding the same set $\Sigma$.
With these observation, the proof follows by theorem~\ref{theo:LRulesPreserveJT}.
\end{proof}

When the target variables are not mutually independent and distinct from the variables occurring in $\Sigma$, we have to apply theorem~\ref{theo:LRulesPreserveJT} and so take into account the changes in $\Sigma$.

Let us now deal with Almost Trust.
Also in this case, we start with a lemma of preservation, but only for construction and restricted only to some logical constants:

\begin{lemma}[Construction with Implication, Conjunction and Disjunction preserves Inequality]
\label{lemma:LRulesPreserveLBound}
If $\textgoth{a}_{1}:\beta^{1} _{f_{1}} , \ldots , \textgoth{a}_{n}: \beta^{n}_{f_{n}} \mapsto ^{cstr} \textgoth{c}: \delta _{f_{n+1}}$, and $\textgoth{b}_{1}:\beta^{1} _{g_{1}} , \ldots , \textgoth{b}_{n}: \beta^{n}_{g_{n}} \mapsto ^{cstr} \textgoth{d}: \delta _{g_{n+1}}$, and moreover the target variable of $\textgoth{d}$ is identical to the target variable of $\textgoth{c}$, $\delta$ is $ \bot $-free, and for every $ i $ s.t. $1 \leq i \leq n$, $f_{i} \geq g_{i} $, then $f_{n+i} \geq g_{n+i} $.
\end{lemma}

\begin{proof}
The proof follows from the observation that if $ f_{i} \geq g_{i}$ and $ g_{i} \geq g_{j}$, then $ f_{i} +  f_{j} \geq g_{i}+ g_{j}$ and $ f_{i} \times f_{j} \geq g_{i} \times g_{j}$.    
\end{proof}

Note that we required that $\delta$ contains only $ \rightarrow $, $ + $ and $ \times $, but not $ \bot $.
This requirement is needed, since we can easily find counterexamples to preservation of inequality under negation.
As an example, let us assume that $\textgoth{a}: \beta _{f}$ and $\textgoth{b}: \beta _{g}$ with $ f\gneq g $, that is $\textgoth{a}$ is not ET with respect to $\textgoth{b}$ regarding $\beta$.
For $ I \rhd \bot $, $\textgoth{a}: \beta ^{\bot} _{1-f}$ and $\textgoth{b}: \beta ^{\bot} _{1-g}$.
Moreover, it is a trivial arithmetical fact that if $ f\gneq g $ then $ 1-g \gneq 1-f $.

Note also that, as opposed to lemma~\ref{lemma:CompPresAgree}, lemma~\ref{lemma:LRulesPreserveLBound} cannot be generalized for logical deconstructions.
Indeed, let us assume that $f_{i}$ and $g_{i}$ are the probabilities of the major premises of two instances of an E-rule, and $f_{j}$ and $g_{j}$ are the probabilities of their minor premises.
Moreover, let us assume that $ f_{i} \geq g_{i}$ and $ f_{j} \geq g_{j}$.
The probabilities associated with the conclusion of the instances of the E-rule are respectively $ \frac{f_{i}}{f_{j}} $ and $ \frac{g_{i}}{g_{j}}$ for $\times$, and $ f_{i} - f_{j} $ and $ g_{i} - g_{j}$ for $+$.
However, $ f_{i} \geq g_{i}$ and $ f_{j} \geq g_{j}$ entails neither $ \frac{f_{i}}{f_{j}} \geq \frac{g_{i}}{g_{j}}$, nor $ f_{i} - f_{j} \geq g_{i} - g_{j}$.
An easy counterexample for both of them is the following: $ f_{i} = 0.6 $,  $ g_{i} = 0.6 $, $ f_{j} = 0.3 $ and $ g_{j} = 0.2 $.

With these results, we can now address AT preservation:

\begin{theorem}[Construction with Implication, Conjunction and Disjunction preserves AT]
\label{theorem:LRulesPreserveAT}
Let us consider two sets $\textgoth{A}$ and $\textgoth{B}$ of applied ML systems and assume that for every $\textgoth{a}_{i} \in \textgoth{A}$ there is $\textgoth{b}_{i} \in \textgoth{B}$ such that $\textgoth{a}_{i}$ is AT with respect to $\textgoth{b}_{i}$, regarding the set of values $ B^{i} $.
If $\textgoth{a}_{1}, \ldots \textgoth{a}_{n} \mapsto ^{cstr} \textgoth{c}$, and $\textgoth{b}_{1}, \ldots \textgoth{b}_{n} \mapsto ^{cstr} \textgoth{d}$, and moreover the target variable of $\textgoth{d}$ is identical to the target variable of $\textgoth{c}$, then for every value $\delta$ in the closure of $ \bigcup^{n}_{i = 1} B^{i} $ under $ \rightarrow $, $ + ^{\Bot}$ and $ \times $, if $\textgoth{c}: \delta _{f}$ and $\textgoth{d}: \delta _{g}$, then $ f \geq g $.
Hence, $ \textgoth{c} $ is AT with respect to $\textgoth{d}$, regarding all values $\delta$ in the closure of $ \bigcup^{n}_{i = 1} B^{i} $ under $ \rightarrow $, $ + ^{\Bot}$ and $ \times $.\footnote{
The closure of the set $B$ under $ + ^{\Bot}$ is the set $B'$ s.t.: $B \subseteq B'$ and if $\beta \in B'$, $\delta \in B'$, and $\beta ~ \Bot  ~ \delta$, then $\beta + \delta \in B'$.
Note that this restriction on the closure under $+$ could be removed.
In fact, for $\textgoth{c}: \delta _{f}$ to be possible, $\delta$ cannot contain a disjunction of values that are not mutually exclusive, due to restrictions on the rules of disjunction.
Nevertheless, we chose to make this restriction explicit to prevent erroneous applications of the theorem.
}
\end{theorem}

\begin{proof}
From definition~\ref{def:ATNoAtm}, $\textgoth{a}_{i}$ is AT with respect to $\textgoth{b}_{i}$, regarding the set of values $ B^{i} $, iff, for every value $ \beta \in B^{i} $, if $ \textgoth{a}:\beta _{f_{i}} $ and $ \textgoth{b}:\beta _{g_{i}} $, then $ f_{i}\geq g_{i} $.
In the same way, $ \textgoth{c} $ is AT with respect to $\textgoth{d}$, regarding all values $\delta$ in the closure of $ \bigcup^{n}_{i = 1} B^{i} $ under $ \rightarrow $, $ + ^{\Bot}$ and $ \times $, iff, if $ \textgoth{c}:\delta _{f} $ and $ \textgoth{d}:\delta _{g} $, then $ f\geq g $.
Hence, the result follows from lemma~\ref{lemma:LRulesPreserveLBound}.
Indeed, note that, even though we consider all $\delta$ in the closure of $ \bigcup^{n}_{i = 1} B^{i} $ under $ \rightarrow $, $ + ^{\Bot}$ and $ \times $, the assumptions $\textgoth{a}_{1}, \ldots \textgoth{a}_{n} \mapsto ^{cstr} \textgoth{c}$ and $\textgoth{b}_{1}, \ldots \textgoth{b}_{n} \mapsto ^{cstr} \textgoth{d}$, together with $\textgoth{c}: \delta _{f}$ and $\textgoth{d}: \delta _{g}$, make sure that we can focus only on the values that can be obtained by applications of the I-rules. 
\end{proof}

Also in this case, when the target variables are mutually independent and do not occur in $\Sigma$, we can obtain a stronger result, regarding the preservation of General AT:

\begin{corollary}[Construction with Implication, Conjunction and Disjunction preserves General AT]
\label{PreservationLogicATGeneral}
Given two ML systems dubbed respectively the \textit{original} and the \textit{copy}, if the \textit{copy} is Almost Trustworthy with respect to the \textit{original} regarding a set of lists of values assignments $\Sigma$, a set of mutually independent atomic target variables $\mathscr{A}$ that do not occur in $\Sigma$ and a function $v$ assigning relevant values to each atomic target variable $a\in \mathscr{A}$, then the \textit{copy} is Almost Trustworthy with respect to the \textit{original} also regarding the same set of lists of values assignments $\Sigma$, the set of target variables $\mathscr{T}$, obtained closing $ \mathscr{A} $ under conjunction, and the function $v'$ from $t\in \mathscr{T}$ to outputs so defined:
\begin{itemize}
    \item for every atomic target variable $a$ occurring in $t$, $v'(a)$ is the closure of $v(a)$ under $+^{\Bot}$;
    \item $v'(\langle t,u \rangle )$ is the closure under $+^{\Bot}$ of $v'(t) \times v'(u)$.
\end{itemize}
\end{corollary}

\begin{proof}
Like for corollary~\ref{PreservationLogicJTGeneral}, the proof just follows from theorem~\ref{theorem:LRulesPreserveAT}, observing that, when the target values of the premises of a rule do not occur in the antecedent of any of the other premises, the antecedent remains the same in the conclusion and the only I-rule for conjunction that can be applied is I$\times ^{\ci}$.
The only tricky part is the definition of $v'$, which we have provided for induction on the complexity of $t\in \mathscr{T}$.
We need to close any induction step under $+ ^{\Bot}$ but not under $\bot$, due to the lemma~\ref{lemma:LRulesPreserveLBound}.
\end{proof}

As for Equal Trust, the following preservation result can be proved:

\begin{theorem}[Logical Construction preserves ET]
\label{theorem:LRulesPreserveET}
Let us consider two sets $\textgoth{A}$ and $\textgoth{B}$ of applied ML systems and assume that for every $\textgoth{a}_{i} \in \textgoth{A}$ there is $\textgoth{b}_{i} \in \textgoth{B}$ such that $\textgoth{a}_{i}$ is ET with respect to $\textgoth{b}_{i}$, regarding the set of values $ B^{i} $.
If $\textgoth{a}_{1}, \ldots \textgoth{a}_{n} \mapsto ^{cstr} \textgoth{c}$, and $\textgoth{b}_{1}, \ldots \textgoth{b}_{n} \mapsto ^{cstr} \textgoth{d}$, and moreover the target variable of $\textgoth{d}$ is identical to the target variable of $\textgoth{c}$, then for every value $\delta$ in the closure of $ \bigcup^{n}_{i = 1} B^{i} $ under $ \rightarrow $, $ + ^{\Bot} $, $ \times $ and $ \bot $, if $\textgoth{c}: \delta _{f}$ then $\textgoth{d}: \delta _{f}$.
Hence, $ \textgoth{c} $ is ET with respect to $\textgoth{d}$, regarding all values $\delta$ in the closure of $ \bigcup^{n}_{i = 1} B^{i} $ under $ \rightarrow $, $ + ^{\Bot} $, $ \times $ and $ \bot $.
\end{theorem}

\begin{proof}
The proof is like the one for theorem~\ref{theo:LRulesPreserveJT}.
We only have to focus on the values $ B^{i} $, since only for them we know that the probabilities in the premises are the same for $\textgoth{a}_{i}$ and $\textgoth{b}_{i}$.
\end{proof}

Like for theorem~\ref{theo:LRulesPreserveJT}, also theorem~\ref{theorem:LRulesPreserveET} holds for logical deconstruction as well:

\begin{theorem}[Logical deconstruction preserves ET]
\label{theorem:LERulesPreserveET}
Let us consider two sets $\textgoth{A}$ and $\textgoth{B}$ of applied ML systems and assume that for every $\textgoth{a}_{i} \in \textgoth{A}$ there is $\textgoth{b}_{i} \in \textgoth{B}$ such that $\textgoth{a}_{i}$ is ET with respect to $\textgoth{b}_{i}$, regarding the set of values $ B^{i} $.
If $\textgoth{a}_{1}, \ldots \textgoth{a}_{n} \mapsto ^{decstr} \textgoth{c}$, and $\textgoth{b}_{1}, \ldots \textgoth{b}_{n} \mapsto ^{decstr} \textgoth{d}$, and moreover the target variable of $\textgoth{d}$ is identical to the target variable of $\textgoth{c}$, then for every value $\delta$ that is a sub-value of an element $\beta$ in $ \bigcup^{n}_{i = 1} B^{i} $, if $\textgoth{c}: \delta _{f}$ then $\textgoth{d}: \delta _{f}$.
Hence, $ \textgoth{c} $ is ET with respect to $\textgoth{d}$, regarding all values $\delta$ that are sub-values of an element $\beta$ in $ \bigcup^{n}_{i = 1} B^{i} $.\footnote{
A sub-value of $ \beta $ is recursively defined in the following way: $ \beta $ is a sub-value of itself; $ \beta $ is a sub-value of $ \beta ^{\bot}$; $ \beta $ and $ \gamma $ are sub-values of $ \beta \rightarrow \gamma$, $ \beta + \gamma$ and $ \beta \times \gamma$.
}
\end{theorem}
        
\begin{proof}
    The theorem follows from lemma~\ref{lemma:CompPresAgree}.
    Note that when we apply right E-rules, the value of the target variable in the conclusion is a sub-value of the value of the target variable in one of the premises.
\end{proof}

Also in this case, when the target variables are mutually independent and do not occur in $\Sigma$, we can obtain a stronger result, regarding the preservation of General ET:

\begin{corollary}[Logical Construction preserves General ET]
\label{PreservationLogicETGeneral}
Given two ML systems dubbed respectively the \textit{original} and the \textit{copy}, if the \textit{copy} is Equally Trustworthy with respect to the \textit{original} regarding a set of lists of values assignments $\Sigma$, a set of mutually independent atomic target variables $\mathscr{A}$ that do not occur in $\Sigma$ and a function $v$ assigning relevant values to each atomic target variable $a\in \mathscr{A}$, then the \textit{copy} is Equally Trustworthy with respect to the \textit{original} also regarding the same set of lists of values assignments $\Sigma$, the set of target variables $\mathscr{T}$, obtained closing $ \mathscr{A} $ under conjunction, and the function $v'$ from $t\in \mathscr{T}$ to outputs so defined:
\begin{itemize}
    \item for every atomic target variable $a$ occurring in $t$, $v'(a)$ is the closure of $v(a)$ under $+^{\Bot}$ and $\bot $;
    \item $v'(\langle t,u \rangle )$ is the closure under $+^{\Bot}$ and $\bot$ of $v'(t) \times v'(u)$.
\end{itemize}
\end{corollary}

\begin{proof}
    The proof is like the one for corollary~\ref{PreservationLogicATGeneral}.
    Note that in this case function $v'$ is constructed closing under negation as well.
\end{proof}

Like inequality, also the zeros in the distribution of probabilities are preserved when the logical construction uses only implication, conjunction, and disjunction.
This result will constitute the missing ingredient of WT:

\begin{lemma}[Construction with Implication, Conjunction and Disjunction preserves Zero]
\label{lemma:LRulesPreserveZero}
If $\textgoth{a}_{1}:\beta^{1} _{f_{1}} , \ldots , \textgoth{a}_{n}: \beta^{n}_{f_{n}} \mapsto ^{cstr} \textgoth{c}: \delta _{f_{n+1}}$, and $\textgoth{b}_{1}:\beta^{1} _{g_{1}} , \ldots , \textgoth{b}_{n}: \beta^{n}_{g_{n}} \mapsto ^{cstr} \textgoth{d}: \delta _{g_{n+1}}$, and moreover the target variable of $\textgoth{d}$ is identical to the target variable of $\textgoth{c}$, $\delta$ is $ \bot $-free, and for every $ i $ s.t. $1 \leq i \leq n$, $f_{i} = 0$ iff $g_{i} = 0$, then $f_{n+i} = 0$ iff $g_{n+i} = 0$.
\end{lemma}

\begin{proof}
Let us assume that $ f_{i} = 0 $ iff $ g_{i} = 0$ and $ f_{j} = 0 $ iff $ g_{j} = 0$.
$ f_{i} +  f_{j} = 0 $ iff $ f_{i} = 0 $ and $ f_{j} = 0 $.
But then $ g_{i} = 0 $ and $ g_{j} = 0 $, and so $ g_{i} +  g_{j} = 0 $.
Moreover, $ f_{i} \times  f_{j} = 0 $ iff $ f_{i} = 0 $ or $ f_{j} = 0 $.
But then $ g_{i} = 0 $ or $ g_{j} = 0 $, and so $ g_{i} \times  g_{j} = 0 $.
\end{proof}

Lemma~\ref{lemma:LRulesPreserveZero} cannot be generalized for logical deconstructions, at least for disjunction.
Indeed, let us assume that $ f_{i} = 0 $ iff $ g_{i} = 0$ and $ f_{j} = 0 $ iff $ g_{j} = 0$, where $f_{i}$ and $ g_{i}$ are the probabilities of the disjunction occurring in the major premise of the respective applications of the elimination rule for $+$ and $f_{j}$ and $ g_{j}$ are the probabilities of the disjunct occurring in their minor premise.
Hence, the probabilities of the conclusion of the applications of the E-rule are $ f_{i} - f_{j} $ and $ g_{i} - g_{j} $.
However, it is not possible to prove $ f_{i} - f_{j} = 0 $ iff $ g_{i} - g_{j} = 0$. 
An easy counterexample is the following: $ f_{i} = 0.6 $,  $ g_{i} = 0.6 $, $ f_{j} = 0.6 $ and $ g_{j} = 0.2 $.
Nevertheless, a partial result is still provable about E-rules:

\begin{lemma}[Deconstruction with Implication and Conjunction preserves Zero]
\label{lemma:LERulesPreserveZero}
If $\textgoth{a}_{1}:\beta^{1} _{f_{1}} , \ldots , \textgoth{a}_{n}: \beta^{n}_{f_{n}} \mapsto ^{decstr} \textgoth{c}: \delta _{f_{n+1}}$, and $\textgoth{b}_{1}:\beta^{1} _{g_{1}} , \ldots , \textgoth{b}_{n}: \beta^{n}_{g_{n}} \mapsto ^{decstr} \textgoth{d}: \delta _{g_{n+1}}$, and moreover the target variable of $\textgoth{d}$ is identical to the target variable of $\textgoth{c}$, $\delta$ is $ \bot $-free and $ + $-free, and for every $ i $ s.t. $1 \leq i \leq n$, $f_{i} = 0$ iff $g_{i} = 0$, then $f_{n+i} = 0$ iff $g_{n+i} = 0$.
\end{lemma}

\begin{proof}
    The proof follows from the observation that if we assume that $ f_{i} = 0 $ iff $ g_{i} = 0$, we can easily prove $ \frac{f_{i}}{f_{j}} = 0 $ iff $ \frac{g_{i}}{g_{j}} = 0$, since the application of E-rules for $\times$ require that $ f_{j} \neq 0 $ and $ g_{j} \neq 0$.
\end{proof}

As opposed to what happens with theorem~\ref{theorem:LRulesPreserveAT}, preservation of zero applies to negation as well, both for construction and deconstruction:

\begin{lemma}[Construction and Deconstruction with Negation Preserve Zero]
\label{lemma:NegPreserveZero}
Let us consider two ML systems $\textgoth{a}$ and $\textgoth{b}$.
If for every $\bot$-free value $ \beta $, $\textgoth{a} : \beta _{0}$ iff $\textgoth{b} : \beta _{0}$, then for every value $ \delta $, $\textgoth{a} : \delta _{0}$ iff $\textgoth{b} : \delta _{0}$.
\end{lemma}

\begin{proof}
Let us assume that for every possible $\bot$-free value $ \beta $ of $\textgoth{a} $ and $ \textgoth{b}$, $\textgoth{a} : \beta _{0}$ iff $\textgoth{b} : \beta _{0}$.
By induction on the complexity of $ \delta $, let us prove that $\textgoth{a} : \delta _{0}$ iff $\textgoth{b} : \delta _{0}$.
The cases in which $ \delta = \beta_{1} + \beta_{2} $ or $ \delta = \beta_{1} \times \beta_{2} $ are proved as for lemma~\ref{lemma:LRulesPreserveZero}.
For $ \delta = \beta ^{\bot} $, let us assume that $\textgoth{a} : \delta _{0}$ and $\textgoth{b} : \delta _{r}$, with $ r\gneq 0 $.
By the rules for negation, $\textgoth{a} : \beta _{1}$ and $\textgoth{b} : \beta _{1-r}$, with $ r\gneq 0 $.
Since $ 1-r \neq 1 $, there has to be some $\bot$-free value $ \gamma $ s.t. $ \gamma ~ \Bot ~ \beta$ and $\textgoth{b} : \gamma _{s}$ with $ r \geq s \gneq 0 $.
However, since $ \gamma $ and $ \beta$ are mutually exclusive and $\textgoth{a} : \beta _{1}$, it follows that $\textgoth{a} : \gamma _{0}$, against induction hypothesis.\footnote{
The availability of the $\bot$-free value $ \gamma $ s.t. $ \gamma ~ \Bot ~ \beta$ follows from the observation that $\bot$ can be defined using $+$ in TNDPQ.
See the $*$-transformation in~\ref{appendix}.
}
\end{proof}

With these lemmas in place, we can prove preservation of Weak Trust when logical construction uses only implication, conjunction and disjunction:

\begin{theorem}[Construction with Implication, Conjunction and Disjunction preserve WT]
\label{theorem:LRulesPreserveWT}
Let us consider two sets $\textgoth{A}$ and $\textgoth{B}$ of applied ML systems and assume that for every $\textgoth{a}_{i} \in \textgoth{A}$ there is $\textgoth{b}_{i} \in \textgoth{B}$ such that $\textgoth{a}_{i}$ is WT with respect to $\textgoth{b}_{i}$, regarding the set of values $ B^{i} $.
If $\textgoth{a}_{1}, \ldots \textgoth{a}_{n} \mapsto ^{cstr} \textgoth{c}$, and $\textgoth{b}_{1}, \ldots \textgoth{b}_{n} \mapsto ^{cstr} \textgoth{d}$, and moreover the target variable of $\textgoth{d}$ is identical to the target variable of $\textgoth{c}$, then for every value $\delta$ in the closure of $ \bigcup^{n}_{i = 1} B^{i} $ under $ \rightarrow $, $ + ^{\Bot} $ and $ \times $, if $\textgoth{c}: \delta _{f}$ and $\textgoth{d}: \delta _{g}$, then $ f \geq g $.
Moreover, for every possible value $\beta$, $\textgoth{c}: \beta _{0}$ iff $\textgoth{d}: \beta _{0}$.
Hence, $ \textgoth{c} $ is WT with respect to $\textgoth{d}$, regarding all values $\delta$ in the closure of $ \bigcup^{n}_{i = 1} B^{i} $ under $ \rightarrow $, $ + ^{\Bot} $ and $ \times $.
\end{theorem}

\begin{proof}
As seen in subsection~\ref{subsec:entailment}, WT corresponds to AT plus the requirement that the systems give probability zero to the same values.
Hence, if $\textgoth{a}_{i}$ is WT with respect to $\textgoth{b}_{i}$ regarding the set of values $ B^{i} $, $\textgoth{a}_{i}$ is AT with respect to $\textgoth{b}_{i}$ regarding the same values.
Moreover, since $\delta$ is in the closure of $ \bigcup^{n}_{i = 1} B^{i} $ under $ \rightarrow $, $ + ^{\Bot} $ and $ \times $, we can apply theorem~\ref{theorem:LRulesPreserveAT} and obtain that $ \textgoth{c} $ is AT with respect to $\textgoth{d}$ regarding values $\delta$.
As for preservation of the zeros, if $\textgoth{a}_{i}$ is WT with respect to $\textgoth{b}_{i}$, then for every value $ \beta $, $\textgoth{a}_{i}: \beta_{0}$ iff $\textgoth{b}_{i}: \beta_{0}$.
Hence, by lemma~\ref{lemma:LRulesPreserveZero} and lemma~\ref{lemma:NegPreserveZero}, for every possible value $\beta$, $\textgoth{c}: \beta _{0}$ iff $\textgoth{d}: \beta _{0}$.
In conclusion, since $ \textgoth{c} $ is AT with respect to $\textgoth{d}$ regarding values $\delta$ in the closure of $ \bigcup^{n}_{i = 1} B^{i} $ under $ \rightarrow $, $ + ^{\Bot} $ and $ \times $, and for every possible value $\beta$, $\textgoth{c}: \beta _{0}$ iff $\textgoth{d}: \beta _{0}$, $ \textgoth{c} $ is WT with respect to $\textgoth{d}$, regarding all values $\delta$ in the closure of $ \bigcup^{n}_{i = 1} B^{i} $ under $ \rightarrow $, $ + ^{\Bot} $ and $ \times $.
\end{proof}

Note that, since AT is not preserved under logical deconstruction, and AT is needed for WT, WT is not preserved under logical deconstruction either.

Let us now address General WT, which requires the usual restrictions on $\Sigma$ and $\mathscr{A}$:

\begin{corollary}[Logical Construction preserves General WT]
\label{PreservationLogicWTGeneral}
Given two ML systems dubbed respectively the \textit{original} and the \textit{copy}, if the \textit{copy} is Weakly Trustworthy with respect to the \textit{original} regarding a set of lists of values assignments $\Sigma$, a set of mutually independent atomic target variables $\mathscr{A}$ that do not occur in $\Sigma$ and a function $v$ assigning relevant values to each atomic target variable $a\in \mathscr{A}$, then the \textit{copy} is Weakly Trustworthy with respect to the \textit{original} also regarding the same set of lists of values assignments $\Sigma$, the set of target variables $\mathscr{T}$, obtained closing $ \mathscr{A} $ under conjunction, and the function $v'$ from $t\in \mathscr{T}$ to outputs so defined:
\begin{itemize}
    \item for every atomic target variable $a$ occurring in $t$, $v'(a)$ is the closure of $v(a)$ under $+^{\Bot}$;
    \item $v'(\langle t,u \rangle )$ is the closure under $+^{\Bot}$ of $v'(t) \times v'(u)$.
\end{itemize}
\end{corollary}

\begin{proof}
    The result follows from corollary~\ref{PreservationLogicATGeneral} and lemma~\ref{lemma:LRulesPreserveZero}
\end{proof}

\begin{example}[Logical Preservation of Trust under Logical Construction]
\label{exa:AtGenJT}
Let us consider a pair of diagnostic systems devised for the detection of some kinds of diseases such as \textit{Chickenpox} and \textit{Hepatitis}.
The systems consider a different variable for each disease and give the probabilities for its severity, with possible outputs $Absent$, $Minor$, $Moderate$, $Major$, and $Extreme$.
Let us moreover assume that the systems are JT one another when a Test Set $\Sigma$ of patients is considered and the target variable is either \textit{Chickenpox} or \textit{Hepatitis}.
Hence, for any $\sigma \in \Sigma$, the ML systems give the same probabilities for these diseases.
As an example:
\medskip
\[
\sigma \rhd Chickenpox : Absent _{0.2}, Minor_{0.4}, Moderate_{0.3}, Major_{0.1}, Extreme_{0}
\]

\medskip
\noindent
and 
\medskip
\[
\sigma \rhd Hepatitis : Absent _{0.5}, Minor_{0.3}, Moderate_{0.2}, Major_{0}, Extreme_{0}
\]

\medskip
By corollary~\ref{PreservationLogicJTGeneral} the systems are JT also regarding logically complex judgments  constructed from these atomic judgments.
As an example, 
\medskip
\[
\sigma \rhd \langle Chickenpox , Hepatitis \rangle : Absent \times (Minor + Moderate) _{0.1}
\]

\medskip
Let us now assume that one system (copy) is AT with respect to the other (original) when a Test Set $\Sigma$ of patients is considered, the target variable is either \textit{Chickenpox} or \textit{Hepatitis}, and the relevant values are $Major$, and $Extreme$.
That is, we know that the copy does not underestimate the most severe cases of \textit{Chickenpox} and \textit{Hepatitis}.
By corollary~\ref{PreservationLogicATGeneral} the copy is AT with respect to the original, also regarding logically complex variables $t$ constructed closing $\{ Chickenpox ,$ $ Hepatitis \}$ under conjunction and implication, when the relevant outputs are assigned to $t$ by $v$ as follows:
\begin{itemize}
    \item $v(Chickenpox) $ and $ v(Hepatitis)$ are the closure of $\{ Major, Extreme \}$ under $+^{\Bot}$;
    \item $v(\langle t,u \rangle )$ is the closure under $+^{\Bot}$ of $v(t) \times v(u)$.
\end{itemize}

As an example, the copy is AT with respect to the original, regarding the variable $\langle Chickenpox ,$ $ Hepatitis \rangle$, and the relevant output:
\medskip
\[
((Major + Extreme) \times Extreme) + (Extreme \times Major)
\]

\medskip
\noindent
That is, the copy is AT with respect to the original, regarding the pair of conditions \textit{Chickenpox} and \textit{Hepatitis}, considering the relevant output ``either \textit{Chickenpox} is $Major$ or $Extreme$ and \textit{Hepatitis} is $Extreme$, or \textit{Chickenpox} is $Extreme$ and \textit{Hepatitis} is $Major$''.
Nonetheless, the copy is not necessarily AT with respect to the original, regarding the variable $\langle Chickenpox ,$ $ Hepatitis \rangle$, if the relevant output is:
\medskip
\[
(Major + Extreme) \times Extreme ^{\bot}
\]

\medskip
\noindent
That is, the copy is not necessarily AT with respect to the original, regarding the pair of conditions \textit{Chickenpox} and \textit{Hepatitis}, considering as relevant output ``\textit{Chickenpox} is $Major$ or $Extreme$, and \textit{Hepatitis} is not $Extreme$''.

Note that in both cases we are assuming that the probabilities of having \textit{Chickenpox} or \textit{Hepatitis} are mutually independent.
This is not a necessary assumption.
However, when atomic variables are not mutually independent, changes in $ \sigma $ have to be considered as well.
\end{example}

\section{Conclusions}
\label{sec:Conclusion}

In this paper, we have developed the system TNDPQ to analyze how logically complex queries are evaluated by ML systems.
Then, we have defined different notions of copy of an ML system, based on which element of the original system is changed in the copy, and proposed different notions of trustworthiness for such copies.
After a discussion of the relations between these different notions of trustworthiness and their possible combinations, we have investigated whether and to what extent trust about atomic queries is preserved when logically complex queries are considered. Future research shall be devoted to the implementation of a trustworthiness preservation module within the MIRAI Toolbox (\url{https://mirai.systems}) a platform for evaluation and risk assessment of ML systems.

In our investigation, we have not discussed specific properties of the original system that we want to preserve in the copy.
In contrast, we kept the investigation as general as possible, focusing only on the preservation of the output.
Clearly, this gives some benefits.
However, when specific epistemological or ethical properties of the original system are considered, special precautions could be necessary to preserve them in the copy.
For this reason, there are many possible questions that we have left unanswered.
As an example, ML systems are usually required to be fair in their judgments, ignoring protected attributes or even balancing biased Data Sets.
As observed in the literature on intersectionality, fair judgments about some protected properties do not necessarily remain fair when intersections between these properties are considered.
This kind of phenomenon suggests that further investigation is required to generalize the results of this paper when the desiderata is the preservation of less simple properties.

\section*{Acknowledgments}

\paragraph*{Funding:} This work was supported by the Project ``SMARTEST: Simulation of Probabilistic Systems for the Age of the Digital Twin'' (MUR - PRIN 2022 - PRIN202223GPRIM\_01 - CUP:G53D23008030006).

\appendix 
\renewcommand{\thesection}{\Alph{section}} % corrected redefinition of '\thesection'
\makeatletter
\renewcommand\@seccntformat[1]{\appendixname\ \csname the#1\endcsname.\hspace{0.5em}}
\makeatother

\section{}
\label{appendix}

In section~\ref{sec:ProofSystem}, we used mutual exclusivity of outputs in the formulation of the I-rules for disjunction.
In this appendix, we will clarify this concept and define formal criteria for exclusivity.

\begin{definition}[Generalized disjunction]
\label{def:Gen+}
    Using $+$, we define a generalized notion of disjunction as follows:
    \medskip
    \[
    \beta^{i} = \bigoplus_{j \in \{i\}} \beta^{j}
    \]
    \[
    \bigoplus_{j \in J} \beta^{j} + \beta^{k} = \beta^{k} + \bigoplus_{j \in J} \beta^{j} = \bigoplus_{j \in J \cup \{k\}} \beta^{j}
    \]
    
    \medskip
    Note that generalized disjunction identifies disjunctions up to associativity, commutativity and idempotence. 
\end{definition}

Recall that $I$ is the set of all indexes for the possible atomic values of $a$.\footnote{
For simplicity, we assume that the set of atomic values for different atomic variables are always mutually disjoint.}

\begin{definition}[Mutual Exclusivity]
    Given a possibly non-atomic variable $t$, outputs $\beta$ and $\delta$ are mutually exclusive for $t$, denoted $t: \beta ~ \Bot ~ t: \delta$, iff $\beta \times \delta \rightarrow \bot$ is provable in any classical propositional system
    %(let us say a natural deduction system)
    that, for every atomic $a$ occurring in $t$, contains the axioms:
    \begin{description}
    \item[Ax. of exclusivity] $ \forall i,j \in I (i\neq j \rightarrow \alpha _{i} \times \alpha ^{j}\rightarrow \bot)$.
    \item[Ax. of exhaustivity] $\bigoplus_{i \in I} \alpha _{i}$
    \end{description}
We call ``classical theory of $\beta$ and $\delta$'' such a theory. Below we will denote it simply with $T$.
\end{definition}

Note that mutual exclusivity is defined, relying on a classical system expressing the mutual incompatibility of all attributions of atomic values to variables.
In other words, the system formalizes the consequences of the assumption that we are working only with incompatible atomic outputs.

\begin{example}[Mutually Exclusive Values]
Le us consider an ML system that deals with Data Sets containing information about age distributions in a population.
    Values $x \leq 20$ and $x \leq 30$ are not mutually exclusive for the atomic variable \textit{Age}, and so are not acceptable.
    On the contrary, values $x \leq 20$, $ 20 \lneq x \leq 30$, and $x \gneq 30$ are mutually exclusive and exhaustive for the atomic variable \textit{Age}, and so are acceptable.
    Hence, we can extend a classical system with the following axioms:
    \begin{description}
    \item[Ax.1] $(x \leq 20 \times 20 \lneq x \leq 30) \rightarrow \bot$
    \item[Ax.2] $(x \leq 20 \times x \gneq 30) \rightarrow \bot$
    \item[Ax.3] $(20 \lneq x \leq 30 \times x \gneq 30) \rightarrow \bot$
    \end{description}
    \noindent
    and use this set of axioms to decide mutual exclusivity of logically complex values.
    As an example, we can easily prove $Age: x \leq 20 + x \gneq 30  ~ \Bot ~ Age: 20 \lneq x \leq 30$.
\end{example}

\begin{definition}[$ \beta  ^{*}$ and $ \beta  ^{\infty}$]
    Given a value $\beta$ for an atomic variable $a$, let us define $\beta ^{*}$ as the result of applying  one of the following transformations (called $*$-transformations and abbreviated with $\rightharpoonup ^{*}$) to $\beta$:
\begin{enumerate}
    \item substitute every positive atomic value $\alpha ^{i}$ in $\beta$ with $\bigoplus_{j \in \{i\}} \alpha ^{j}$;
    \item substitute every negated atomic value $(\alpha ^{i})^{\bot}$ in $\beta$ with $\bigoplus_{j \in I \setminus \{i\}} \alpha ^{j}$;
    \item substitute every negated generalized disjunction of positive atomic values $(\bigoplus_{j \in J  } \alpha ^{j})^{\bot}$ in $\beta$ with $\bigoplus_{j \in \overline{J}} \alpha ^{j}$.
\end{enumerate}
Where $I$ is the set of all the indexes for the possible atomic values of $a$, $J$ is a set of some indexes for the possible atomic values of $a$, and $\overline{J}$ is its complement.
Moreover, let us define $\beta  ^{\infty}$ as the result of applying as many times as possible these transformations.
\end{definition}

Before going on, let us see an example of $*$-transformation.

\begin{example}
Let us see the reduction of $(((\alpha ^{1} + \alpha ^{2})^{\bot} + \alpha ^{3})^{\bot} + \alpha ^{4})^{\bot}$, step by step up to $((((\alpha ^{1} + \alpha ^{2})^{\bot} + \alpha ^{3})^{\bot} + \alpha ^{4})^{\bot})^{\infty}$:

\bigskip

\scalebox{0.90}{$
   \begin{aligned}
    (((\alpha ^{1} + \alpha ^{2})^{\bot} + \alpha ^{3})^{\bot} + \alpha ^{4})^{\bot} \rightharpoonup ^{*} ((\bigoplus_{j \in I\setminus \{1,2\}} \alpha ^{j} + \alpha ^{3})^{\bot} + \alpha ^{4})^{\bot} \rightharpoonup ^{*} \qquad \quad \\
    ((\bigoplus_{j \in I\setminus \{1,2\}} \alpha ^{j})^{\bot} + \alpha ^{4})^{\bot} \rightharpoonup ^{*}     \rightharpoonup ^{*} (\bigoplus_{j \in \{1,2\}} \alpha ^{j} + \alpha ^{4})^{\bot} \rightharpoonup ^{*} (\bigoplus_{j \in \{1,2,4\}} \alpha ^{j})^{\bot} \rightharpoonup ^{*} \bigoplus_{j \in I\setminus \{1,2,4\}} \alpha ^{j}
    \end{aligned}
    $}
\bigskip

\end{example}

\begin{lemma}[$\beta ^{\infty}$ is a generalized disjunction of atoms]
\label{lemma:*Disj}
    For every value $\beta$ of an atomic variable $a$, $\beta ^{\infty}$ is a generalized disjunction of positive atomic values.
\end{lemma}

\begin{proof}
Obvious, by inspecting the definition of $ * $-transformation.
\end{proof}

\begin{lemma}[$*$-transformation preserves Exclusivity]
\label{lemma:disjunctionstar}
    For every pair of values $\beta$ and $\delta$ of an atomic variable $a$, $a: \beta ~ \Bot ~ a: \delta$ iff $a: \beta ^{\infty} ~ \Bot ~ a: \delta ^{\infty}$.
\end{lemma}

\begin{proof}
    Let us assume that we have a proof of $\beta \times \delta \rightarrow \bot$ in the classical theory of $\beta $ and $ \delta $.
    We can prove $\beta ^{\infty} \times \delta ^{\infty} \rightarrow \bot$ by changing each formula $\gamma$ in this proof with $\gamma ^{\infty}$.
    All the applications of the rules (and the axioms) remain valid, apart from those for negation (\textit{non contradiction} and \textit{classical reductio}), which can be shown to be derivable as follows.

For non contradiction, with $\mid I\setminus \{i\}\mid = n-1$ and the obvious generalization of $E+$ to a generalized disjunction:

\begin{prooftree}
    \AxiomC{$ \bigoplus_{j \in I \setminus \{i\}} \alpha ^{j} $}

    \AxiomC{$ \alpha ^{i} $}  
    \AxiomC{$ [\alpha ^{1} ] $}  
    \BinaryInfC{$\bot$}

    \AxiomC{$ \ldots n-1 ~~~times $} 

    \AxiomC{$ \alpha ^{i} $}  
    \AxiomC{$ [\alpha _{n}] $}  
    
    \BinaryInfC{$\bot$}

\QuaternaryInfC{$\bot$}
\end{prooftree}

For \textit{classical reductio}:

\begin{prooftree}
    \AxiomC{$ \bigoplus_{i \in I} \alpha ^{i} $}

    \AxiomC{$ [\bigoplus_{j \in I \setminus \{k\}} \alpha ^{j} ] $}  
    \noLine
    \UnaryInfC{$\vdots$}
    \noLine
    \UnaryInfC{$\bot$}
    \UnaryInfC{$\alpha ^{k}$}

    \AxiomC{$ [\alpha ^{k} ] $} 

\TrinaryInfC{$\alpha ^{k}$}
\end{prooftree}

Now, let us assume that $\beta ^{\infty} \times \delta ^{\infty} \rightarrow \bot$, and prove that $\beta \times \delta \rightarrow \bot$.
It is convenient to prove the stronger result that $\Gamma \vdash _{T} \phi$ iff $\Gamma \vdash _{T} \phi [\gamma / \gamma ^{*}]$.
The proof is by induction on the complexity of $\phi$ and by cases on its outmost logical connective.
\begin{description}
\item[Base:] If $\phi$ is an atom, then $\phi = \gamma = \alpha ^{i}$, and the proof is obvious.

\item[Step (negated atom):] If $\phi$ is a negated atom, then $\phi = (\alpha ^{i}) ^{\bot}$, and $*$-transformation can be applied only to $(\alpha ^{i}) ^{\bot}$, so $\gamma = (\alpha ^{i}) ^{\bot}$.
If $\Gamma \vdash _{T} (\alpha ^{i}) ^{\bot}$ and $ \vdash _{T} \bigoplus_{i \in I} \alpha ^{i}$ by Exhaustivity, we obtain $\Gamma \vdash _{T} \bigoplus_{j \in I \setminus \{i\}} \alpha ^{j}$ by disjunctive syllogism. If $\Gamma \vdash _{T} \bigoplus_{j \in I \setminus \{i\}} \alpha ^{j}$ and $ \alpha ^{i} , \alpha ^{j} \vdash _{T} \bot$ for every $i,j \in I$ s.t. $ i\neq j$, we obtain $\Gamma , \alpha ^{i} \vdash _{T} \bot$ by various applications of $\lor$E, and so $\Gamma \vdash _{T} (\alpha ^{i})^{\bot}$ by introduction of negation.

\item[Step (negated disjunction of atoms):] If $\phi$ is a negated disjunction of atoms, then $\phi = (\bigoplus_{j \in J  } \alpha ^{j})^{\bot}$, and so $\gamma = (\bigoplus_{j \in J  } \alpha ^{j})^{\bot}$.
If $\Gamma \vdash _{T} (\bigoplus_{j \in J  } \alpha ^{j})^{\bot}$, then $\Gamma \vdash _{T} \bigoplus_{j \in \overline{J}} \alpha ^{j}$ follows from $ \vdash _{T} \bigoplus_{i \in I} \alpha ^{i}$ by disjunctive syllogism.
On the other direction, assuming $\Gamma \vdash _{T} \bigoplus_{j \in \overline{J}} \alpha ^{j}$, from   $ \alpha ^{i} , \bigoplus_{j \in \overline{J}} \alpha ^{j} \vdash _{T} \bot$ for every $i\in J$, we obtain $\Gamma , \bigoplus_{j \in \overline{J}} \alpha ^{j} , \bigoplus_{i \in J  } \alpha ^{i} \vdash _{T} \bot$ by various applications of $\lor$E, and so $\Gamma , \bigoplus_{j \in \overline{J}} \alpha ^{j} \vdash _{T} (\bigoplus_{j \in J  } \alpha ^{j})^{\bot}$ by introduction of negation.
Hence, if $\Gamma \vdash _{T} \bigoplus_{j \in \overline{J}} \alpha ^{j}$, we obtain the desired conclusion by transitivity of deduction.

\item[Step disjunction] If $\phi = \psi + \psi '$, then by induction hypothesis $\Gamma \vdash _{T} \psi$ iff $\Gamma \vdash _{T} \psi [\gamma / \gamma ^{*}]$, and $\Gamma \vdash _{T} \psi '$ iff $\Gamma \vdash _{T} \psi ' [\gamma / \gamma ^{*}]$.
Hence, assuming $\Gamma \vdash _{T} \psi + \psi '$, we use $\psi \vdash _{T} \psi [\gamma / \gamma ^{*}] +\psi ' [\gamma / \gamma ^{*}]$ and $\psi ' \vdash _{T} \psi [\gamma / \gamma ^{*}] +\psi ' [\gamma / \gamma ^{*}]$ to derive the conclusion;

\item[Steps conjunction and conditional] Similar to the previous case.
\end{description}

This concludes the proof that $\Gamma \vdash _{T} \phi$ iff $\Gamma \vdash _{T} \phi [\gamma / \gamma ^{*}]$, from which it easily follows that  if $\beta ^{\infty} \times \delta ^{\infty} \rightarrow \bot$, then $\beta \times \delta \rightarrow \bot$.
\end{proof}

\begin{lemma}[Exclusivity for disjunctions of atoms]
\label{lemma:EcludDisj}
    Given an atomic variable $a$, $a: \bigoplus_{j \in J } \alpha ^{j} ~ \Bot ~ a: \bigoplus_{k \in K } \alpha ^{k}$ holds iff $ K \cap J = \emptyset $. 
\end{lemma}

\begin{proof}
$\bigoplus_{j \in J } \alpha ^{j},\bigoplus_{k \in K } \alpha ^{k}\vdash_{I}\bot$ only using the axiom of exclusivity, and this requires that $J \cap K = \emptyset$:

\begin{prooftree}
\footnotesize
\def\defaultHypSeparation{\hskip .1in}
    \AxiomC{$ \bigoplus_{k \in K} \alpha ^{k} $}

        \AxiomC{$ \bigoplus_{j \in J} \alpha ^{j} $}

            \AxiomC{$ [\alpha _{j^{1}}]^{1} $}  
            \AxiomC{$ [\alpha _{k^{1}}]^{2} $}  
            \BinaryInfC{$\bot$}

    \AxiomC{$ \ldots   \mid J \times K \mid ~~~times$} 
    
            \AxiomC{$ [\alpha _{j_{n}}]^{1} $}  
            \AxiomC{$ [\alpha _{k_{n}}]^{2} $}  
            \BinaryInfC{$\bot$}

\RightLabel{1}
\QuaternaryInfC{$\bot$}

\RightLabel{2}
\BinaryInfC{$\bot$}
\end{prooftree}

\noindent
Where $J \cap K = \emptyset$ entails that $\bot$ is derivable from any combination of $\alpha ^{j} $ and $\alpha ^{k} $.
On the contrary, if $J \cap K \neq \emptyset$, then let us assume $i\in J \cap K$.
The valuation giving True to $\alpha ^{i}$ and False to every other atom is a model of $\bigoplus_{j \in J } \alpha ^{j} \times \bigoplus _{k \in K } \alpha ^{k}$.
Hence, since $\bigoplus_{j \in J } \alpha ^{j} \times \bigoplus_{k \in K } \alpha ^{k}$ has a model it cannot be inconsistent.    
\end{proof}

\begin{theorem}[Decision (atomic variables)]
\label{teo:decAtom}
Given an atomic variable $a$ and two outputs $\beta$ and $\delta$, there is an effective procedure for deciding whether $a: \beta ~ \Bot ~ a: \delta$.
\end{theorem}

\begin{proof}
By lemma~\ref{lemma:disjunctionstar}, $a: \beta ~ \Bot ~ a: \delta$ iff $a: \beta ^{\infty}~ \Bot ~ a: \delta ^{\infty}$.
\end{proof}

Now that we have a decision procedure for the mutual exclusivity of outputs of the atomic variables, we need to generalize this result to non-atomic variables.

\begin{theorem}[Decision]
Given a variable $u$ and two outputs $\beta$ and $\delta$, there is an effective procedure for deciding whether $u: \beta ~ \Bot ~ u: \delta$.
\end{theorem}

\begin{proof}
The proof is by induction on the complexity of $u$.

\begin{description}
\item[Base] For atomic variables, we have theorem~\ref{teo:decAtom}.
\item[Step (conjunction)] If $u = \langle t,v \rangle$, then we define a criterion for $\langle t,v \rangle: \beta ~ \Bot ~ \langle t,v \rangle: \delta$ by induction on the $\times$-complexity of $\beta$ and $\delta$ ($\times C(\beta) + \times C(\delta)$), where $\times$-complexity is defined as follows:\footnote{Note that all conjunctions are of complexity $1$, regardless of the complexity of their subvalues.
Apart from this, the definition is as usual.
Moreover, note that $+$ in the argument of $\times C$ stands for disjunction, while in the \textit{definiens} it stands for sum.}
\begin{equation*}
\begin{split}
\times C(\gamma \times \gamma ') & := 1  \\
\times C(\gamma + \gamma ') & := \times C(\gamma ) +  \times C(\gamma')\\
\times C(\gamma ^{\bot}) & := \times C(\gamma ) +  1\\
\end{split}
\end{equation*}

\begin{description}
\item[Base (conjunction)] The base case is when both $\beta$ and $\delta$ are conjunctions.
Let us assume $\beta = \gamma ^{1}\times \gamma ^{2}$ and $\delta = \gamma ^{3}\times \gamma ^{4}$.
In this case, $\langle t,v \rangle: \gamma ^{1}\times \gamma ^{2} ~ \Bot ~ \langle t,v \rangle: \gamma ^{3}\times \gamma ^{4}$ iff $ t: \gamma ^{1} ~ \Bot ~ t: \gamma ^{3}$ or $v :\gamma ^{2} ~ \Bot ~ v : \gamma ^{4}$.
Indeed, mutual exclusivity is defined only for values of the same variable, so if $\langle t,v \rangle: \gamma ^{1}\times \gamma ^{2} ~ \Bot ~ \langle t,v \rangle: \gamma ^{3}\times \gamma ^{4}$ then $ t: \gamma ^{1} ~ \Bot ~ t: \gamma ^{3}$ or $v :\gamma ^{2} ~ \Bot ~ v : \gamma ^{4}$.
Moreover, if $ \gamma ^{1} \times \gamma ^{3} \rightarrow \bot$, then clearly $ (\gamma ^{1} \times \gamma ^{3}) \times (\gamma ^{2} \times \gamma ^{4}) \rightarrow \bot$, and the same holds for $\gamma ^{2} \times \gamma ^{4}  \rightarrow \bot$.

\item[Step] By cases on the outmost connective in $\beta$:

\begin{description}
\item[Double negation] $\beta = \gamma ^{\bot \bot} $ is treated as $\beta = \gamma $;
\item[Disjunction] If $\beta = \bigoplus_{j \in J } \gamma ^{j} $, then $\langle t,v \rangle: \bigoplus_{j \in J } \gamma ^{j} ~ \Bot ~ \langle t,v \rangle: \delta$ iff for every $j \in J$, $\langle t,v \rangle: \gamma ^{j} ~ \Bot ~ \langle t,v \rangle: \delta$, since, by distribution of conjunction over disjunction, $ \bigoplus_{j \in J } \gamma ^{j} \times \delta \rightarrow \bot $ iff $ \bigoplus_{j \in J } (\gamma ^{j} \times \delta) \rightarrow \bot $, and this is equivalent to $ \bigotimes_{j \in J } (\gamma ^{j} \times \delta \rightarrow \bot) $;\footnote{
Where $\bigotimes$ is the obvious generalization of conjunction dual to the generalization of disjunction seen in definition~\ref{def:Gen+}.
}

\item[Negated conjunction] $\beta = (\gamma ^{1} \times \gamma ^{2} )^{\bot} $ is treated as $\beta = ((\gamma ^{1}) ^{\bot} \times \gamma ^{2} ) + (\gamma ^{1}  \times (\gamma ^{2}) ^{\bot}) +  ((\gamma ^{1}) ^{\bot} \times (\gamma ^{2}) ^{\bot}) $;
\item[Negated disjunction] If $\beta = (\bigoplus_{j \in J } \gamma ^{j}) ^{\bot} $, then we have $\langle t,v \rangle: (\bigoplus_{j \in J } \gamma ^{j})^{\bot} ~ \Bot ~ \langle t,v \rangle: \delta$ iff for some $j \in J$, $\langle t,v \rangle: \gamma ^{j} ~ \Bot ~ \langle t,v \rangle: \delta$;
\end{description}

\end{description}

\item[Step (conditional)] If $u = [t]v $, then we define a criterion for $[t]v: \beta ~ \Bot ~ [t]v: \delta$ by induction on the $\rightarrow$-complexity of $\beta$ and $\delta$ ($\rightarrow C(\beta) + \rightarrow C(\delta)$), where $\rightarrow$-complexity is defined as follows:
\begin{equation*}
\begin{split}
\rightarrow C(\gamma \rightarrow \gamma ') & := 1  \\
\rightarrow C(\gamma + \gamma ') & := \rightarrow C(\gamma ) +  \rightarrow C(\gamma')\\
\rightarrow C(\gamma ^{\bot}) & := \rightarrow C(\gamma ) +  1\\
\end{split}
\end{equation*}

\begin{description}
\item[Base (conditional)] The base case is when both $\beta$ and $\delta$ are conditionals.
Let us assume $\beta = \gamma ^{1}\rightarrow \gamma ^{2}$ and $\delta = \gamma ^{3}\rightarrow \gamma ^{4}$.
In this case, $[t]v: \gamma ^{1}\rightarrow \gamma ^{2} ~ \Bot ~ [t]v: \gamma ^{3}\rightarrow \gamma ^{4}$ iff $\gamma ^{1} = \gamma ^{3}$ and $v :\gamma ^{2} ~ \Bot ~ v : \gamma ^{4}$. Indeed, evaluating a conditional $\sigma \rhd [t]v: \gamma  \rightarrow \gamma ' $ is the same as evaluating the consequence, adding the antecedent to the premise: $\sigma , t: \gamma \rhd v: \gamma ' $.

\item[Step] By cases on the outmost connective in $\beta$:

\begin{description}

\item[Double negation] $\beta = \gamma ^{\bot \bot} $ is treated as $\beta = \gamma $;

\item[Disjunction] If $\beta = \bigoplus_{j \in J } \gamma ^{j} $, then $[t]v: \bigoplus_{j \in J } \gamma ^{j} ~ \Bot ~ [t]v: \delta$ iff for every $j \in J$, $[t]v: \gamma ^{j} ~ \Bot ~ [t]v: \delta$;

\item[Negated disjunction] If $\beta = (\bigoplus _{j \in J } \gamma ^{j}) ^{\bot} $, then we have $[t]v: (\bigoplus_{j \in J } \gamma ^{j})^{\bot} ~ \Bot ~ [t]v: \delta$ iff for some $j \in J$, $[t]v: \gamma ^{j} ~ \Bot ~ [t]v: \delta$;

\item[Negated conditional] $\beta = (\gamma ^{1} \rightarrow \gamma ^{2} )^{\bot} $ is treated as $\beta = \gamma ^{1} \rightarrow (\gamma ^{2}) ^{\bot} $, since

\begin{prooftree}
\AxiomC{$ \sigma \rhd [t]u: (\gamma ^{1} \rightarrow \gamma ^{2})^{\bot} _{g} $}
	\RightLabel{{\tiny I/E$\rhd\bot$}}
	\doubleLine
	\UnaryInfC{$ \sigma \rhd [t]u: (\gamma ^{1} \rightarrow \gamma ^{2}) _{1-g} $}
	\RightLabel{{\tiny I/E$\rhd\rightarrow$}}
	\doubleLine
	\UnaryInfC{$ \sigma , t: \gamma ^{1} \rhd u: (\gamma ^{2}) _{1-g} $}
	\RightLabel{{\tiny I/E$\rhd\bot$}}
	\doubleLine
\UnaryInfC{$ \sigma , t: \gamma ^{1} \rhd u: (\gamma ^{2})^{\bot} _{g} $}
\RightLabel{{\tiny I/E$\rhd\rightarrow$}}
	\doubleLine
	\UnaryInfC{$ \sigma \rhd [t]u: (\gamma ^{1} \rightarrow (\gamma ^{2})^{\bot}) _{g} $}
\end{prooftree}
\end{description}

\end{description}

\end{description}
\end{proof}

\begin{example}
    Let us check whether $\langle a_{1},a_{2} \rangle: (\alpha ^{1} \times \alpha ^{3})^{\bot} ~ \Bot ~ \langle a_{1},a_{2} \rangle: (\alpha ^{1} + \alpha ^{2}) \times \alpha ^{3} $.
    The variable $\langle a_{1},a_{2} \rangle$ is a conjunction, so for the first value we use the clause for negated conjunctions.
    Hence, $(\alpha ^{1} \times \alpha ^{3})^{\bot}$ is treated like $((\alpha ^{1})^{\bot} \times \alpha ^{3}) + (\alpha ^{1} \times (\alpha ^{3})^{\bot}) + ((\alpha ^{1})^{\bot} \times (\alpha ^{3})^{\bot})$.
    We then apply the clause for disjunctions $\langle a_{1},a_{2} \rangle: ((\alpha ^{1})^{\bot} \times \alpha ^{3}) + (\alpha ^{1} \times (\alpha ^{3})^{\bot}) + ((\alpha ^{1})^{\bot} \times (\alpha ^{3})^{\bot}) ~ \Bot ~ \langle a_{1},a_{2} \rangle: (\alpha ^{1} + \alpha ^{2}) \times \alpha ^{3} $ iff $\langle a_{1},a_{2} \rangle: ((\alpha ^{1})^{\bot} \times \alpha ^{3}) ~ \Bot ~ \langle a_{1},a_{2} \rangle: (\alpha ^{1} + \alpha ^{2}) \times \alpha ^{3} $, and $\langle a_{1},a_{2} \rangle: (\alpha ^{1} \times (\alpha ^{3})^{\bot}) ~ \Bot ~ \langle a_{1},a_{2} \rangle: (\alpha ^{1} + \alpha ^{2}) \times \alpha ^{3} $, and moreover $\langle a_{1},a_{2} \rangle: ((\alpha ^{1})^{\bot} \times (\alpha ^{3})^{\bot}) ~ \Bot ~ \langle a_{1},a_{2} \rangle: (\alpha ^{1} + \alpha ^{2}) \times \alpha ^{3} $.

    At this point, both values are of $\times $complexity $1$, so we can apply the base case for conjunctions to all these combinations, and check whether all of them hold.
    Let us start with the first one.
    $\langle a_{1},a_{2} \rangle: ((\alpha ^{1})^{\bot} \times \alpha ^{3}) ~ \Bot ~ \langle a_{1},a_{2} \rangle: (\alpha ^{1} + \alpha ^{2}) \times \alpha ^{3} $ iff $a_{1} : (\alpha ^{1})^{\bot}  ~ \Bot ~ a_{1} : \alpha ^{1} + \alpha ^{2}$ or $a_{2} : \alpha ^{3}  ~ \Bot ~ a_{2} : \alpha ^{3}$.
        Clearly, $a_{2} : \alpha ^{3} $ and $ a_{2} : \alpha ^{3}$ are not mutually exclusive.
        As for $a_{1} : (\alpha ^{1})^{\bot}  ~ \Bot ~ a_{1} : \alpha ^{1} + \alpha ^{2}$, we apply the criterion for atomic variables: $((\alpha ^{1})^{\bot})^{\infty} = \bigoplus _{j\in I \setminus \{ 1\}} \alpha ^{j}$, and since $(I \setminus \{ 1\}) \cap \{ 1,2\} = \{ 2\} \neq \emptyset$, $(\alpha ^{1})^{\bot}  $ and $ \alpha ^{1} + \alpha ^{2}$ are not mutually exclusive values for $a_{1}$.
    Since mutual exclusivity does not hold already for the first pair of conjunctions, we can stop here and conclude that $(\alpha ^{1} \times \alpha ^{3})^{\bot} $ and $(\alpha ^{1} + \alpha ^{2}) \times \alpha ^{3} $ are not mutually exclusive values for the variable $\langle a_{1},a_{2} \rangle$.

\end{example}
            
\bibliographystyle{elsarticle-harv} 
  \bibliography{biblio}
\end{document}